\documentclass[12pt]{amsart}

\usepackage{amssymb,amsfonts,amscd,graphics,mathrsfs,color,
  mathtools,enumitem,centernot}
\usepackage[arrow]{xy}

\theoremstyle{plain}
\newtheorem{theorem}{Theorem}[section]
\newtheorem{corollary}[theorem]{Corollary}
\newtheorem{definition}[theorem]{Definition}
\newtheorem{lemma}[theorem]{Lemma}
\newtheorem{proposition}[theorem]{Proposition}
\newtheorem{example}[theorem]{\it Example}
\newtheorem{remark}[theorem]{\it Remark}
\newtheorem{conjecture}[theorem]{\it Conjecture}

\DeclarePairedDelimiter\floor{\lfloor}{\rfloor}

\def\hd{{\hat d}}
\def\ifbox{\hbox{if \ }}
\def\hot{\hat\j}
\def\iot{j}
\def\Gal{\mathrm{Gal}}
\def\lmt{\longmapsto}

\def\e{\mathbf{e}}
\def\bfp{\mathbf{p}}
\def\bfq{\mathbf{q}}
\def\bfr{\mathbf{r}}
\def\proj{\P\Zbd^2}
\def\prop{\P\Z_\bp^2}
\def\OK{\O_\K}
\def\QD{\Q(\sqrt D)}

\def\A{\mathcal{A}}
\def\D{\mathbf{D}}
\def\E{\mathbb{E}}
\def\F{\mathbb{F}}
\def\G{\mathcal{G}}
\def\K{\mathbb{K}}
\def\N{\mathbb{N}}
\def\T{\mathbb{T}}
\def\rN{\mathrm{N}}
\def\O{\mathcal{O}}
\def\P{\mathbb{P}}
\def\Q{\mathbb{Q}}
\def\R{\mathbb{R}}
\def\Z{\mathbb{Z}}
\def\sF{\mathscr{F}}
\def\sT{\mathscr{T}}

\def\SO{\mathrm{SO}}
\def\SU{\mathrm{SU}}
\def\U{\mathrm{U}}
\def\GL{\mathrm{GL}}
\def\SL{\mathrm{SL}}
\def\bC{{\bar C}}
\def\bd{{\bar d}}
\def\bp{{\bar p}}
\def\Zd{\Z_d}
\def\Zbd{\Z_\bd}
\def\2{\sqrt2}
\def\3{\sqrt3}
\def\5{\sqrt5}
\def\sqd{\sqrt{d\+1}}
\def\u{\\[-6pt]}
\def\half{{\ts\frac12}}

\renewcommand{\,}{\kern1pt}
\renewcommand{\!}{\kern-1pt}
\def\Ph{\Phi}
\def\id{\mathbf1}

\def\bk#1#2{\left<#1,#2\right>}
\def\kb#1#2{#1\kern1pt #2^*}
\def\bmk#1#2#3{#1^*#2\kern1pt #3}
\def\zzz{\zz\kern1pt\zz^*}

\def\diag{\mathrm{diag}}
\def\th{\theta}
\def\mod#1{\ (\mathrm{mod}\ #1)}
\def\qbox#1{\quad\hbox{#1}\quad}
\def\x{x^{}}
\def\z{z^{}}
\def\al{\alpha^{}}
\def\oz{\ol z}
\def\ow{\ol w}
\def\Pa{P_\alpha}
\def\Pb{P_\beta}

\def\SS{\mathbb{S}}

\def\sD{\mathscr{D}}

\def\sphi{\sqrt\varphi}
\def\tr{\mathop{\mathrm{tr}}}
\def\ti{\!\times\!}

\def\ch{\kern1pt\raise2pt\hbox{$\chi$}}

\def\ge{\geqslant}
\def\le{\leqslant}
\def\y{\\[3pt]}
\def\yy{\\[5pt]}
\def\yyy{\\[8pt]}
\def\cd{\hbox{\large$\cdot$}}
\def\suml{\sum\limits}

\def\Im{\mathop{\mathrm{Im}}}
\def\Re{\mathop{\mathrm{Re}}}
\def\fa{\mathfrak{a}}

\def\fm{\mathfrak{m}}
\def\fn{\mathfrak{n}}
\def\fp{\mathfrak{p}}
\def\ft{\mathfrak{t}}
\def\fu{\mathfrak{u}}
\def\ga{\gamma}
\def\de{\delta}
\def\la{\lambda}
\def\om{\omega}
\def\si{\sigma}
\def\De{\Delta}

\def\op{\oplus}

\def\ts{\textstyle}
\def\ba{\begin{array}}
\def\ea{\end{array}}
\def\be#1{\begin{equation}\label{#1}}
\def\ee{\end{equation}}
\def\+{\!+\!}
\def\-{\!-\!}
\def\={\!=\!}

\def\sb#1{\kern-5pt\lower5pt\hbox{$_{#1}$}}
\def\k{\kern1pt}
\def\n{\kern-1pt}
\def\hra{\hookrightarrow}
\def\lra{\longrightarrow}
\def\ol{\overline}
\def\C{\mathbb{C}}
\def\CP{\mathbb{CP}}
\def\CPD{\CP^{d-1}}

\def\gl{\mathfrak{gl}}

\def\su{\mathfrak{su}}

\def\ww{\mathbf{w}}
\def\xx{\mathbf{x}}
\def\zz{\mathbf{z}}
\def\rw{[\ww]}
\def\rz{[\zz]}
\def\vs{\vskip10pt}

\def\Unit{\mathcal{U}}
\def\Cl{\mathcal{C}}

\begin{document}
\parskip1pt
\parindent16pt
\mathsurround.5pt

\title[Moment maps and Galois orbits]{\large
  Moment maps and Galois orbits\yy in quantum information theory}

\author{Kael Dixon and Simon Salamon}

\maketitle


\vskip30pt

\begin{quote}\small 
  {\bf Abstract.} SIC-POVMs are configurations of points or rank-one
  projections arising from the action of a finite Heisenberg group on
  $\C^d$. The resulting equations are interpreted in terms of moment maps by
  focussing attention on the orbit of a cyclic subgroup and the maximal torus
  in $\U(d)$ that contains it. The image of a SIC-POVM under the associated
  moment map lies in an intersection of real quadrics, which we describe
  explicitly. We also elaborate the conjectural description of the related
  number fields and describe the structure of Galois orbits of overlap phases.
\end{quote}

\section{Introduction}

It is conjectured that, for every positive integer $d,$ there exist $d^2$
points in the complex projective space $\CPD$ that are pairwise
equidistant with respect to the standard Fubini-Study metric. This is
equivalent to asserting that, as an adjoint orbit in the Lie algebra
$\su(d)\cong\R^{d^2-1}$ of Killing vector fields, projective space contains the
$d^2$ vertices of a regular simplex.  The common distance of separation depends
only on $d$ and the diameter of $\CPD$. Such sets of points were
originally studied under the guise of {\it equiangular lines}
\cite{LS,DGS,Hogg}, and are related to topics such as {\it tight frames} in
design theory.

Versions of the conjecture arose from the pioneering work of Gerard Zauner
\cite{Z}, who revitalized the subject from the viewpoint of quantum
information. Such a set of $d^2$ points defines a so-called {\it symmetric
  informationally complete positive operator measure}, for short SIC-POVM. The
concept of a POVM in quantum theory was introduced in \cite{Dav,Hol,DL}, though
in the present context it is an entirely discrete object. The rank-one
projections defined by the $d^2$ points provide an optimal way to measure a
mixed state, and applications of SIC-POVMs arise in {\it quantum tomography}, a
topic advanced by Ugo Fano \cite{ugo}\footnote{On a historical note that places
  the discipline in a family context, his brother Robert worked on Shannon-Fano
  coding \cite{claude}, whilst their father was the algebraic geometer Gino
  Fano.}.

The advent of computing has emphasized the importance of finite-dimensional
Hilbert spaces in quantum theory, and in particular metric properties of
$\CPD$. Virtually all known SIC-POVMs arise as orbits of a discrete Heisenberg
group, acting as $\Zd^2=\Zd\times\Zd$ on $\CPD$ ($\Zd$ denotes $\Z/d\Z$
throughout the paper). A vector $\zz\in\C^d$ is called \emph{fiducial} if it
has unit norm and the orbit $\Zd^2\cdot\rz$ is a SIC-POVM. Fiducial vectors
have been constructed exactly for all $d\le21$ and in a few higher
dimensions. There is strong numerical evidence for their existence for all $d$
up to a 3-figure value, and seminal papers on the subject include
\cite{App,Bend,Fla,Grl,RSC,SG,Zhu}.  For the remainder of the paper, we work
with Heisenberg SIC-POVMs.

The existence of SIC-POVMs for $d=2$ and $3$ is elementary. For $d=2,$ $\CP^1$
coincides with the round 2-sphere, any SIC-POVM is an inscribed tetrahedron,
and any two are related by an element of $\SO(3)\cong \SU(2)/\Z_2$. For $d=3,$
there is a one-parameter family of SIC-POVMs on $\CP^2$ up to the action of the
isometry group, any one of which is isometric to the orbit $\Zd^2\cdot\rz$ for
a suitable fiducial vector $\zz$. It is much harder to prove that there are no
other SIC-POVMs on $\CP^2,$ and both known proofs involve an element of
computation \cite{Sz,HS}. It is believed that there are only finitely many
isometry classes of SIC-POVMs in dimensions $d>3,$ and there is strong evidence
for this at least when $d\le50$ \cite{SG}.

The condition for a unit vector $\zz\in\C^d$ to be fiducial is that, for any
non-identity matrix $g\in\Zd^2<\U(d)$ induced from the Heisenberg group, the
complex number $\bmk\zz g\zz$ has norm $\frac1\sqd$. These complex numbers are
referred to as the \emph{(unnormalized) overlap phases}, and can be used to
recover the fiducial vector. These are usually encoded into a $\bd\times\bd$
matrix, called the \emph{overlap phase matrix}, where we follow notation of
Appleby \cite{App} to define \be{bd}\bd = \left\{\begin{matrix} d & \text{if
}d\text{ is odd} \\ 2d & \text{ if }d\text{ is even}.
\end{matrix}\right.\ee
We shall denote this matrix by $\Ph_\zz,$ and interpret it as a map
$\Zbd^2\to\C$. It is not a direct encoding of the numbers $\{\bmk\zz
g\zz:g\in\Zd^2\},$ since there is some algebraic advantage in multiplying these
numbers by $2d$-roots of unity, necessitating the extension from $d$ to
$\bd$. See Section \ref{secSIC}.

One aim of this paper is to extend the moment map approach of \cite{HS} to
arbitrary dimensions. This has led us independently to a formulation involving
the discrete Fourier transform, which is known to be central in understanding
the action of the Heisenberg group \cite{ADF,FHS}, though our emphasis is more
on the underlying geometry. Let us consider the case when $d$ is odd for
simplicity. We fix a cyclic subgroup $C\cong\Zd,$ which (as explained in
Section \ref{secMom}) gives rise to a maximal torus of $\SU(d)$ and a moment
map $\mu^C$ from $\CPD$ onto a standard simplex $\De$ in $\R^{d-1}$. When $C$
is realized as a subgroup of $\Zd^2,$ we relate $\mu^C$ to the map
$\rz\mapsto\Ph_\zz|C,$ and describe the intersection of $\De$ with a product of
circles defined by elements $\Ph_\zz(\bfp)$ with $\bfp\in C$.

Let $\A$ denote the image by $\mu^C$ of points of $\CPD$ whose $C$ orbits
consist of $d$ equidistant points that could hypothetically lie in a
SIC-POVM. Then $\A$ is universal in the sense that it does not depend on $C$
(Lemma \ref{adm}), being a re-interpretation of conditions on the overlap
phases defined by $C$. The set $\A$ is described by Theorem \ref{main} and
(when $d$ is odd) is a Clifford torus, suitably interpreted. This approach will
not help in the quest for SIC-POVMs, but it may conceivably lead to the
construction of an example that is not the orbit of a finite group acting on
$\CPD$.

The other aim of this paper is to build on the conjectural relationship with
number theory developed in
\cite{AYZ,appleby2018constructing,Bengt,Kopp,AFMY}. It is observed
\cite{appleby2018constructing} that all of the known exact examples of
fiducials are equivalent in a certain sense to those having a property that is
called \emph{strongly centred}. See Section \ref{secGalois} for a more precise
definition. Consider such a strongly centred fiducial vector $\zz$ of unit
norm. Then the unnormalized overlap phases generate an abelian extension
$\mathbb E_1$ of the real quadratic field $\K=\QD,$ where $D$ is the
square-free part of $(d-3)(d+1)$. The relevant Galois group to study is the
subgroup $\G$ of $\Gal(\E_1/\Q)$ which preserves the set of overlap phases. The
Galois action of $\G$ on the phases is encoded in a very natural way: for each
$g\in\G,$ there exists some $G_g\in \GL_2\Zbd$ such that $g\circ\Ph_\zz =
\Ph_\zz\circ G_g$. The present work is inspired by the way in which we have
organized the set of overlap phases into a set of moment maps. The form of the
Galois action described above means that it will send moment maps to moment
maps. We will use this structure to describe the orbits of the Galois action.

The element $G_g$ is not unique, being determined only up to left
multiplication by an element of the set
$S = \{G\in\GL_2\Zbd: \Ph_\zz\circ G = \Ph_\zz\}$.
Thus there is some subgroup $M$ of $\GL_2\Zbd$ such that $\G\cong M/S$. We will
focus on $M$ as being more fundamental than $\G,$ since it does not depend on
the fiducial, but only on its type (either $z,a_4,a_6,$ or $a_8,$ see Section
\ref{secGalois}) and $d$. After careful study of the Galois groups of the exact
solutions tabulated in \cite{AYZ, appleby2018constructing}, we observe that
most of these have a property that we call \emph{algebraic}, meaning that
$M\cong \left(\OK(\bd)\right)^\times,$ where $\OK$ is the set of algebraic
integers in $\K$. This condition determines the type:
\begin{proposition}\label{propTypeCriteria}
	An algebraic fiducial is of type
$\left\{\ba{ll}
	a_4 	& \ifbox 3|d \text{ and }D\equiv 1 \mod 3
\\	a_6	& \ifbox d\equiv 3\mod{27}
\\	a_8	& \ifbox 3|d \text{ and }D\equiv 2 \mod 3
\\	z	& \text{otherwise.}
 \ea\right.$
\end{proposition}

\noindent The non-algebraic fiducials are type-$z$ when the above criteria
would suggest type-$a,$ giving a simple modification to the structure of $M$.
It was not previously known whether it is possible to have type-$a_6$
fiducials, but this criterion suggests that these will occur for all $d$
congruent to $3$ modulo $27,$ including the known solution labelled
$30d$.\smallbreak

In \cite{AYZ}, SIC-POVMs are related to natural field extensions called ray
class fields. In particular, it was observed that for every $d$ with known
fiducials, there is a fiducial where $\E_1$ is the ray class field
$\K\big((\bd)\infty_1\big)$ associated to the ideal $(\bd)$ in $\OK$ and the
real embedding $\infty_1$ of $\K$ for which $\sqrt D$ is positive. Moreover,
these are always algebraic and the fixed field of $\G$ is the Hilbert class
field $\K(1)$. We call these \emph{ray class fiducials} and compute their
corresponding symmetry groups:

\begin{theorem}\label{rcf}
  For ray class fiducials, $S$ is isomorphic to the cyclic group $S_0$
  generated by the fundamental unit in $\OK$ modulo $\bd$.
\end{theorem}

More generally, in Theorem \ref{thmExtOfRCF} we show that for algebraic
fiducials, $S$ is a subgroup of $S_0,$ and $\mathcal G$ is a cyclic extension
of that corresponding to the ray class fiducial.  We also describe the slightly
more complicated case of non-algebraic fiducials.

In Section \ref{secGalois}, we discuss the orbits of the Galois action of
$\mathcal G$ on the overlap phases. The main result is Theorem
\ref{thmOrbitFactor}, which shows that these orbits are in one-to-one
correspondence with factors of $\bar d$ in $\mathcal O_\K$.  As a direct
consequence of this and Lemma \ref{lemPrimeOrbit}, we find that

\begin{corollary}\label{corOneOrbit}
  The set of overlap phases are all Galois conjugate if and only if $d>3$ is an
  odd prime congruent to $2$ modulo $3$.
\end{corollary}

\noindent This result relates to a recent work \cite{Kopp} by Kopp, which
studies the case when $d$ is an odd prime congruent to $2$ modulo $3$. Using
the fact that there is only one Galois orbit of overlap phases, the problem of
finding a SIC-POVM is interpreted as finding an algebraic unit satisfying
certain properties. Furthermore, Kopp conjectures that such a unit will come
from Stark's conjecture. Theorem \ref{thmOrbitField} describes in general which
field each Galois orbit takes values in. It may be possible to use this to
generalize Kopp's construction.\smallbreak

The paper is organized as follows. Section \ref{secSIC} is a review of basic
properties of SIC-POVMs, and applies this to the $d=3$ case using results from
\cite{HS}. Section \ref{secMom} generalizes this moment map interpretation to
arbitrary $d,$ which is illustrated by the case $d=4$ and overlap phases
explained in \cite{Bengt}. In Section \ref{secGalois}, we review the relevant
number fields and Galois actions, including carefully motivating the definition
of a strongly centred fiducial. We then introduce the idea of an algebraic
fiducial and compute the group structure of $M$. This allows us compute the
orbits of the Galois action on the overlap phases. In Section \ref{secG}, we
review the conjecture that the number fields associated to a strongly fiducial
are all extensions of ray class fields. We use this to describe the structure
of the field extensions and the fixed fields corresponding to the orbits of
$\G$.\vs

\noindent{\small{\bf Acknowledgements.} The authors are supported by the Simons
  Collaboration on Special Holonomy in Geometry, Analysis, and Physics
  (\#488635 Simon Salamon). They also gratefully acknowledge support from the
  Simons Center for Geometry and Physics, Stony Brook University, at which some
  of the research was presented. We also are grateful for many helpful comments
  and suggestions from two referees.}\vskip20pt

\setcounter{equation}{0}
\section{SIC-POVMs and their overlap phases}\label{secSIC} 

In this section, we follow the approach and notation of \cite{HS}. Complex
projective space $\CPD$ is a compact K\"ahler manifold. Its Riemannian metric
$g$ arises from the standard Hermitian form \be{herm} \bk\ww\zz =
\suml_{i=0}^{d-1}\ow_iz_i\ee on $\C^d$ that is invariant by the unitary group
$\U(d)$. In formulae involving matrices, we shall regard elements of $\C^d$ as
\emph{column} vectors, so that we can identify \eqref{herm} with the matrix
product $\ww^*\zz$. The Hermitian form converts any point $\rw$ of $\CPD$ into
a hyperplane
\[ H_\ww = \{\rz\in\CPD: \bk\ww\zz=0\},\]
and is determined up to a constant by the correspondence $[\ww]\mapsto
H_\ww$.

The Riemannian distance $\de,$ obtained by integrating $g,$ satisfies
\[\cos^2\!\Big({\ts\frac12}\de\big(\rw,\rz\big)\Big) =
\fp\big(\rw,\rz\big),\] where
\[ \fp\big(\rw,\rz\big) = \frac{|\bk\ww\zz|^2}{\|\ww\|^2\|\zz\|^2} =
\frac{\bk\ww\zz\bk\zz\ww}{\bk\ww\ww\bk\zz\zz},\] see \cite[section\,8]{BH}.
Since any two points $\rw,\rz$ are contained in a totally geodesic projective
line $\ell_{\ww,\zz} \cong \CP^1\cong S^2,$ the formula can be proved by
restricting to this 2-sphere. The normalization ensures that the diameter of
$\CPD$ naturally equals $\pi$.

The description above yields

\begin{lemma}
$\fp\big(\rw,\rz\big)$ equals the cross ratio of the four
points $\rw,\rz,\rz',\rw'$ in order, where
$[\ww']=\ell_{\ww,\zz}\cap H_\ww$ and $[\zz']=\ell_{\ww,\zz}\cap H_\zz$.
\end{lemma}

Points $\rw,\rz$ of $\CPD$ represent pure quantum states, and the projective
line they span the set of all their possible superpositions. Observables
correspond to Hermitian matrices whose eigenvalues are the results of
measurements. If $\zz$ is the eigenvector of such a matrix then
$\fp\big(\rw,\rz\big)$ is the probability that the state $[\ww]$ is observed to
coincide with $\rz$. Wigner's theorem states that any isometry $\phi$ of $\CPD$
arises from a unitary or conjugate unitary transformation of $\C^d$. We refer
the reader to proofs by Bargmann \cite{Barg} and Freed \cite{Fr}, the former
mentioning a quaternionic analogue and the latter exploiting the notion of
holonomy (thus both suited to our sponsor). Many more links between quantum
theory and Fubini-Study geometry are explored in \cite{BH}.

Let $\C^{d,d}$ denote the set of complex $d\times d$ matrices, and consider the
subsets
\[\ba{rcl}
\su(d) &=& \left\{A\in\C^{d,d}: A=-A^*,\ \tr A=0\right\},\y
\sD    &=& \left\{A\in\C^{d,d}: A=A^*,\ \tr A =1\right\}.\ea\]
Then $\su(d)$ is the Lie algebra of $SU(d),$ and $\sD$ is the set of
\emph{density matrices}. The map \be{affine}\ts \fa\colon\sD\lra \su(d),\quad
A\mapsto i\left(A-\frac1d\id\right)\ee is an affine isomorphism between these
two spaces of real dimension $d^2-1$.  The Lie algebra has a natural inner
product defined by minus the Killing form:
\[ \left<A,B\right> = -\frac1{2d}\tr(AB),\]
and this enables us to identify $\su(d)$ with its dual $\su(d)^*$.

The mapping $F\colon S^{2d-1}\to\sD$ for which
\be{Them}
F(\zz) = \zzz = \left(\ba{cccc}
|z_0|^2 & \oz_0 z_1 & \oz_0 z_2 & \cdots\\
\oz_1 z_0 & |z_1|^2 & \oz_1 z_2 & \cdots\\
\oz_2 z_0 & \oz_2 z_1 & |z_2|^2 & \cdots\\
\cdots &\cdots &\cdots & \ea\right)\ee
is $\SU(d)$-equivariant and the composition
\be{faf} \fa\circ F\colon\ \CPD \hra \su(d)\cong\su(d)^*\ee
can be identified with the moment mapping defined by the action of $SU(d)$ on
the symplectic manifold $\CPD$ \cite{Kir}. It follows that \eqref{faf}
defines an embedding of $\CPD$ in $\su(d)^*$ as a coadjoint orbit that
is (up to a universal constant) isometric.

Given \eqref{faf}, $\sD$ is the convex hull of the set $\{F(\zz):\zz\in
S^{2d-1}\}$ of pure states. A point of $\sD$ represents a \emph{mixed} quantum
state, whereas an observable is represented by an arbitrary Hermitian matrix
$A$. The expectation that the observable is in the state $\rho\in\sD$ is then
given by $\tr(\rho A),$ so that $\rho$ can be viewed as a probability density
\cite{BH2}.

The acronym SIC-POVM is an abbreviation of {\em symmetric informationally
  complete positive operator valued measure}. In defining this concept, we
adopt the geometric approach of \cite{HS}:

\begin{definition}\label{sic}
  A SIC-POVM is a collection $S=\{[\zz_\alpha]\}$ of $n^2$ points in $\CPD$
  that are mutually equidistant, so that
  \[ \fp(\zz_\alpha,\zz_\beta) = \la,\qquad \alpha\ne\beta\]
  for some fixed $\la\in(0,1)$.
\end{definition}

\noindent The points represent {\it equiangular} 1-dimensional subspaces in
$\C^d$.\smallbreak

The next result expresses the quantity $\la$ in this definition as a
function of $d$:

\begin{lemma} [\rm\cite{DGS}]\label{DGS}
The number of mutually equidistant points possible in $\CPD$ is at most $d^2,$
and if this number is achieved, then $\la=1/(d+1)$.
\end{lemma}

\noindent A quick proof is also given in \cite[section\,3]{HS}.\smallbreak

For fixed $d,$ one therefore knows the distance between any two points
$[\zz_\alpha],[\zz_\beta]\in \CPD$ of a hypothetical SIC-POVM. This innocuous
result helps to make the existence problem tractable, and we record

\begin{definition}\label{cs}
  Two points $\rw,\rz$ of $\CPD$ are \emph{correctly separated} if
  $$\fp\big(\rw,\rz\big)=1/(d+1).$$
\end{definition}

\noindent Thus, any two distinct points of a SIC-POVM are correctly
separated.\smallbreak

If we normalize the representative vectors and set $P_\alpha=F(\zz_\alpha),$
then a SIC-POVM can be equivalently defined as a subset $\{\Pa\}$ of $d^2$
rank-one projections in $\sD$ such that
\[\sum_{\alpha=1}^{d^2}\Pa=d\,\id,\qquad \tr(\Pa\Pb) =
\left\{\ba{ll} 1 \quad & \alpha=\beta\\\la & \alpha\ne\beta,\ea\right.\] This
reflects the etymology of the words \emph{informationally complete} and
\emph{symmetric}. It also allows us to view a SIC-POVM as a regular simplex in
$\su(d)$ whose $d^2$ vertices lie in the coadjoint orbit $\CPD$. It also raises
the question of whether there is an analogue of Lemma \ref{DGS} for other
adjoint orbits, any one of which is a projective variety \cite{Serre} with a
natural K\"ahler metric \cite[Chap.~8]{Besse}.\smallbreak

Fix $d\ge3,$ and set $\om=e^{2\pi i/d}$. Having chosen coordinates on $\C^d,$
one can define two cyclic subgroups $W,H$ of $\U(d),$ with respective
generators
\be{wh} w=\left(\ba{ccccc}
0 & 0 & 0 & \cdots & 1\\
1 & 0 & 0 & \cdots & 0\\
0 & 1 & 0 & \cdots & 0\\
\vdots & \vdots & \vdots && \vdots\\
0 & 0 &  0 & \cdots& 0
\ea\right),\qquad
h=\left(\ba{ccccc}
1 &  0  & 0    & \cdots & 0\\
0 & \om & 0    & \cdots & 0\\
0 & 0  & \om^2 & \cdots & 0\\
\vdots & \vdots & \vdots && \vdots\\
0 & 0 &  0 & \cdots &\om^{d-1}
\ea\right).\ee
Both have order $d,$ and the former acts on column vectors by the shift map
\[ w\,(\z_0,\z_1,\ldots,\z_{d-1})^\top=(\z_{d-1},\z_0,\ldots,\z_{d-2})^\top.\] 
Note that $W=\left<w\right>$ and $H=\left<h\right>$ are subgroups of $\SU(d)$
only if $d$ is odd.

Since $hw=\om\,wh,$ the actions induced by $W,H$ on $\CPD$ commute. A vector
$\zz\in\C^d$ is called \emph{fiducial} if it has unit norm and the orbit
$(W\times H)\cdot\rz$ is a SIC-POVM, as in Definition \ref{sic}. However, $w,h$
generate a subgroup of $\U(d)$ or order $d^3,$ isomorphic to the Heisenberg
group of upper-triangular $3\times3$ matrices defined over the ring $\Z_d$
(with $1$'s on the diagonal). To better understand the resulting phase factors
independently of the parity of $d,$ recall the definition of $\bd$ in
\eqref{bd}, set $\tau=-e^{\pi i/d},$ and define operators
\[ \D_\bfp=\tau^{p_1p_2}w^{p_1}h^{p_2},\quad \bfp=(p_1,p_2).\]
The map \be{full} \D\colon\left\{\ba{ccc} \Zbd^2 &\lra& \U(d)\y \bfp
&\lmt& \D_\bfp\ea\right.\ee is not a homomorphism, but satisfies
\be{DD} \D_\bfp\D_\bfq =
\tau^{\left<\bfp,\bfq\right>}\D_{\bfp+\bfq},\ee where
\[ \langle\bfp,\bfq\rangle = \det(\bfq,\bfp) = q_1p_2-q_2p_1\]
is a symplectic pairing. If $d$ is even, then
\be{eqnDShift} \D_{\bfp+d\bfq} = \bar\tau^{d\langle\bfp,\bfq\rangle}\D_\bfp\D_{d\bfq} =
(-1)^{\langle \bfp,\bfq\rangle} \D_\bfp,\ee 
so that $\D$ is $d$-periodic up to sign.\smallbreak

We next review the concept of overlap phase. Many more details can be found in
\cite{App}.

Let $\zz$ be a unit vector in $\C^d,$ we write $\zz\in S^{2d-1}$. For each
non-zero element $\bfp\in\Zbd^2,$ the quantity
\[ \bmk\zz{\D_\bfp}\zz = \tr(\D_\bfp\zzz)\]
is called an \emph{overlap phase} for the Heisenberg action.

\begin{definition}\label{Phiz}
Let $\zz\in S^{2d-1}$. The \emph{overlap map} associated to $\zz\in S^{2d-1}$
is the mapping $\Ph_\zz\colon\Zbd^2\to\C$ defined by
$\Ph_\zz(\bfp)=\bmk\zz{\D_\bfp}\zz$.
\end{definition}

\noindent Of course, $\Ph_\zz$ can be thought of as a square matrix whose rows
and columns are each indexed by $\Zbd,$ and we shall also refer to its values
as `entries'. Beware that the top left entry $\Ph_\zz(0,0)=1$ is not strictly
speaking an overlap phase. The point is that $\zz$ generates a SIC-POVM if and
only if all the other entries of $\Ph_\zz$ are complex numbers of modulus
$\frac1\sqd$. Shifting the row or column index of an entry by $d$ will at worst
change its sign.\smallbreak

For the moment, let us restrict the index $\bfp$ to the subset $\Zd^2$ if $d$
is even. Then the set $\{\D_\bfp:\bfp\in\Zd^2\}$ forms a unitary basis for the
complex vector space $\gl_d\C$ of complex $d\times d$ matrices, equipped with
the Hermitian product $\left<A,B\right>=(1/d)\tr(A^*B)$. This is true because
\[  \D_\bfp^* = \D_\bfp^{-1} = \tau^{-p_1p_2}h^{-j}w^{-i},\]
so $\D_\bfp^*\D_\bfq$ is a scalar multiple of
$\D_{-\bfp+\bfq},$ and
\[ \tr(\D_\bfp^*\D_\bfq) = \left\{\ba{ll}
d & \hbox{if $\bfp=\bfq\in\Zd^2$}\y
0 & \hbox{otherwise}.\ea\right.\]
It follows that the entries of $\Ph_\zz|\Zd^2$ are essentially the
coordinates of $\zzz$ with respect to this
basis. Indeed, \be{recover} \zzz = 
\frac1d\sum_{\bfp\in\Zd^2}\Ph_\zz(\bfp)\D_{\bfp}^*,\ee
and $\rz$ can be recovered by replacing $\D_{\bfp}^*$ by its first column
in \eqref{recover}, assuming that $z_0\ne0$.

The concept of a unitary basis of operators was introduced in \cite{Sch}, see
also \cite{ADF}. If $\D_\bfp$ were itself Hermitian or skew-Hermitian, then
$\Ph_\zz(\bfp)$ would be a component of the moment map \eqref{faf}. Although
this is not true, we shall explain in the next section precisely how the
overlap map encodes the moment maps associated to maximal tori of the isometry
group of $\CPD$.

\begin{example}\label{d=3}\rm
We summarize the results of \cite{HS,Zhu} for $d=3$. Let $\SS =
\big\{[\zz_\alpha]:1\le\alpha\le9\big\}$ be any SIC-POVM consisting of 9
mutually equidistant points in $\CP^2$. We prove that $\SS$ is in fact
congruent to an orbit of $W\times H$ by means of the following steps.

Up to isometry and re-ordering the points, we may assume that \be{zz12}\ts
[\zz_1]=[0,-1,\om],\qquad [\zz_2]=[0,-1,\om^2]=h\cdot[\zz_1].\ee Having fixed
these points, we may use a residual $\U(1)$ symmetry to take \be{z3} \ts
[\zz_3] = \Big[\cos\phi,\>\cos\big(\phi+\frac{2\pi}3\big),
  \>\cos\big(\phi+\frac{4\pi}3\big)\,\Big],\ee for some angle $\phi$ modulo
$\pi$. Acting by $w$ on $[\zz_3]$ merely has the effect of replacing $\phi$ by
$\phi-\frac{2\pi}3$ on the right-hand side.

It can now be shown that $\SS$ is the orbit of these
  three points under $W,$ so that
  \[\SS = \left\{w^i\cdot[\zz_\alpha]:0\le i\le 2,\ 1\le\alpha\le3\right\}.\]
  The proof of this fact is accomplished in \cite{HS} by a long series of
  lemmas. The first step is to characterize equilateral triangles in $\CP^2$
  with vertices lying on a singular torus generated by points equidistant from
  $[\zz_1]$ and $[\zz_2]$ (thereby extending the circle defined by
  \eqref{z3}). By assumption, $\SS$ will contain $35$ such triangles.

To convert $\SS$ into a visually simpler form, let \be{M}
M={\ts\frac1{\sqrt3}}\!\left(\!\ba{ccc} \om^2 & \om &
1\\1&\om&\om^2\\1&1&1\ea\!\right)\in \U(3).\ee 
$M$ maps $[\zz_3]$ to
$[e^{2i\phi}\om^2,1,0]$ and $\SS$ to the SIC-POVM consisting of the points
\be{33}\ba{ccccc} [0,1,-1] && [0,1,-\om] && [0,1,-\om^2]\yy [1,0,-1] &&
   [1,0,-\om] && [1,0,-\om^2]\yy [e^{2i\phi},1,0] && [e^{2i\phi}\om,1,0] &&
   [e^{2i\phi}\om^2,1,0].\ea\ee
By setting $\phi=(3t+\pi)/2$ and applying the isometry
   $\diag(-e^{-it},-e^{it},1)\in\SU(3)$ we can convert the nine points into the
   orbit $(W\times H)\cdot\rz,$ where $\zz=\frac1{\sqrt2}(0,1,-e^{it})$ and
   $0\le t\le\frac\pi3$.

It is easily verified by hand that $\zz$ is a fiducial vector, and we have
recovered the usual representation of SIC-POVMs for $d=2$. The associated
overlap matrix is
\[ \Ph_\zz = \half\!\left(\!\!\ba{ccc} 
2 & -1 & -1 \\
-e^{-it}&-e^{-it}&-e^{-it}\\
-e^{it}&-e^{it}&-e^{it}
\ea\!\!\right).\]
The moduli space of these SIC-POVMs was described in detail by Zhu \cite{Zhu}
(with the same parameter $t$), building on \cite{App,RSC,Z}. Their congruence
classes can be characterized by the $\U(3)$-invariant triple product
\[ \arg\left[\bk{\zz_1}{\zz_2}\bk{\zz_2}{\zz_3}\bk{\zz_3}{\zz_1}\right]\]
  of three unit vectors in $\C^3$ first defined in \cite{Barg}. For the
  normalized vectors in \eqref{zz12},\eqref{z3}, this argument equals
  $-2\phi+\pi=-3t,$ although its value for many of the 84 triples drawn from
  $\SS$ is independent of $t$. It turns out that, up to isometry, all SIC-POVMs
  for $d=3$ (regarded as an unordered set of nine points) are faithfully
  parametrized by restricting $t$ to lie in the interval $[0,\frac\pi9]$. The
  solutions corresponding to the two endpoints are inequivalent but have
  enhanced symmetries relative to $0<t<\frac\pi9$.
\end{example}\medbreak

The matrix \eqref{M} in the example above belongs to the normalizer $\rN(d)$ of
$W\ti H$ in $\U(d),$ which is known as the \emph{Clifford group}. There is a
representation
\[ U\colon\ \SL_2\Zbd\ltimes\Zd^2 \lra \rN(d)/\U(1)\]
satisfying 
\be{UU} U(F,\bfq)\,\D_\bfp\,U(F,\bfq)^{-1} = 
\om^{\langle\bfq,F\bfp\rangle}\D_{F\bfp}.\ee 
It is an isomorphism if $d$ is odd. This theory can be extended to include
complex conjugation and $\rN(d)$ becomes a subgroup of index 2 in the so-called
\emph{extended Clifford group}. This is the natural symmetry group for
investigating orbits for the action of $\Zbd$ on $\C^d,$ though we shall not
rely on knowledge of its exact structure in this paper. We shall be more
concerned with the symmetries arising from Definition \ref{Phiz}.

\begin{remark}\rm
The Clifford group for $d=5$ plays an essential role in the construction of the
Horrocks-Mumford bundle $\sF,$ which is a stable rank 2 holomorphic vector
bundle over $\CP^4$ \cite{HM}. A generic holomorphic section of $\sF$ vanishes
on a non-singular abelian surface $Z,$ and $\sF$ can be re-constructed from the
normal bundle of $Z$ in $\CP^4$. The surface $Z$ is itself generated by the
tangent lines to a normal elliptic quintic curve invariant by the Heisenberg
group. There is a one-parameter family of such quintic curves in $\CP^4$
(including twelve degenerations to pentagons), which sweep out a surface that
fibres over a modular curve, and a similar phenomenon occurs in higher
dimensions. (All these facts are fully explained in \cite[chapters
  IV,\,VII]{Hulek}.) Although the use of elliptic curves to incorporate points
of a SIC-POVM appears possible only for $d=3$ \cite{hesse,Bend}, the relevance
of these constructions has yet to be fully investigated.
\end{remark}

\setcounter{equation}{0}
\section{Moment maps and quadrics}\label{secMom}

We begin with some notation to handle cyclic subgroups. Fix $d\ge3$. Let
$\bar\pi\colon\Zbd\to\Zd$ denote the natural homomorphism given by
$k+\Zbd\mapsto k+\Zd$ (which is of course the identity if $d$ is odd).

\begin{definition}
Let $\proj$ denote the set of cyclic subgroups of $\Zbd^2$ of size $\bd$. When
$d$ is even, the image of $\bC\in\proj$ under $\bar\pi^2\colon\Zbd^2\to\Zd^2$
is a cyclic subgroup $C$ of $\Zd^2$ of order $d,$ and we say that $C$ and $\bC$
are \emph{associated}.
\end{definition}

For $\bC\in\proj,$ the restriction $\D|\bC$ of \eqref{full} to $\bC$ is a
homomorphism and $d$-periodic even when $d$ is even, as can be seen from
\eqref{DD}. The $d$-periodicity allows us to define a map \be{PhibC}
\Ph_\zz^\bC\colon\ C\lra \C\ee satisfying $\Ph_\zz^\bC\circ\bar\pi^2 =
\Ph_\zz|\bC,$ where $C$ and $\bC$ are associated. We shall refer to
\eqref{PhibC} as the \emph{restricted overlap map} determined by $\bC$.
When $d$ is odd, of course $\bC=C$ and $\Ph_\zz^\bC=\Ph_\zz|C$.

\begin{lemma}
	When $d$ is even, each $C\in\P\Zd^2$ is associated to two
distinct subgroups $\bC,\bC'\in\proj,$ but $\Ph_\zz^\bC$ and $\Ph_\zz^{\bC'}$
agree up to changes of sign on odd entries.
\end{lemma}
\begin{proof}

When $d$ is even, $d\Zbd^2\cong\Z_2^2$. Thus if $\bfp$ is a generator
for $\bC,$ then $\{\bfp+d\bfq:\bfq\in\Zbd^2\}$ has four points. Two of
these (namely $\bfp$ and $\bfp+d\bfp$) lie in $\bC,$ while the other
two lie in a different subgroup $\bC'$ associated to $C$. Choose
$\bfq\in\Zbd^2$ such that $\bfr=\bfp+d\bfq$ generates $\bC'$. In
particular, we have $\bfp\centernot\equiv \bfq\mod 2$. Since neither
element is trivial modulo $2,$ this implies that
$\langle\bfp,\bfq\rangle$ is odd, which can easily be checked by
considering determinants of pairs in $\Z_2^2$.  Using \eqref{DD}, we
have $\D_{k\bfr} =(-1)^k\D_{k\bfp}$ for every $k\in\Zbd,$ thus
\[\Ph_\zz^{\bC'}(k\bfr)=\left\{\ba{cl}
\Ph_\zz^\bC(k\bfp)  & \hbox{if $k$ is even}\y
-\Ph_\zz^\bC(k\bfp)\quad & \hbox{if $k$ is odd},
\ea\right.\] as stated.
\end{proof}
\noindent Note that we shall largely ignore this sign ambiguity as we
shall ultimately be concerned only with the respective images of these
mappings.\smallbreak

Let $n=\floor{d/2}$. Since the restriction of $\D$ to $\bC$ is a homomorphism,
we can identify the image of $\Ph_\zz^\bC$ with an element of the affine space
\be{sT}\ba{rcl} \sT &=& \big\{(\al_0,\ldots,\al_{d-1})\in\C^d\colon \al_0
=1,\ \al_i=\ol{\al_{d-i}}\big\}\yyy &\cong&\left\{\ba{ll} \C^n & \hbox{if $d$
  is odd}\y \C^{n-1}\op\R\ & \hbox{if $d$ is even}.\ea\right.\ea\ee 

We will relate $\Ph_\zz^\bC$ to a moment map. First we need a torus:

\begin{lemma}\label{lemTC}
  For every $C\in\P\Z_d^2,$ there is some maximal torus $\hat{\mathbb
    T}^C$ in $U(d)$ which contains the image of $\bar C$ under $\D$
  for each associated $\bar C\in\P\Z_\bd^2$.
\end{lemma}
\begin{proof}
Let
$\bfp\in \Zbd^2$ generate $\bC$. Since $\gcd(p_1,p_2)\equiv 1\mod\bd,$ there exists some $\bfq\in\Z_\bd^2$ such that $F:=\begin{pmatrix}\bfq & \bfp \end{pmatrix}\in\SL_2\Z_\bd.$ Since $\bfp = F\begin{pmatrix}0 \\ 1\end{pmatrix}$,
 $\D_\bfp$ is conjugate to $\omega^kh$ for some $k\in\Z_d,$ see \eqref{UU}. It
follows that the eigenvalues of $\D_\bfp$ are the distinct $d$\,th roots of
unity $1,\om,\om^2,\ldots,\om^{d-1},$ with an ordered set $(\e_j)$ of unit
eigenvectors indexed by $\Zd$. Note that if $\bfp'\ne\bfp$ but
$\bar\pi^2(\bfp)=\bar\pi^2(\bfp'),$ then $\D_{\bfp'} = -\D_\bfp$. Thus the
unordered set $\{\e_j\}$ of eigenvectors depends only on the subgroup
$C\in\P\Zd^2$.

The corresponding projectors $\{\kb{\e_j}{\e_j}:j\in\Zd\}$ are Hermitian, and
we may identify them with elements of $\fu(d)$. They generate a maximal torus
$\hat \T^C$ in $\U(d)$. Note that this torus is conjugate to the standard
diagonal torus, which contains $\omega^k h,$ so that $\D_\bfp\in \hat\T^C$.
\end{proof}

Note that $\hat\T^C$ descends to a torus $\T^C$
in the projective unitary group $\U(d)/\U(1)$. Recalling \eqref{affine}, its
Lie algebra $\ft^C$ can be identified with the set of elements in the real span
of $\{\kb{\e_j}{\e_j}:j\in\Zd\}$ with trace $1$.

We now can associate to $\bC$ a moment map $\mu^\bC\colon\CPD\to\ft^C$ in
analogy to \eqref{faf} by setting
\[ \mu^\bC_j(\rz):=
\tr\left(\kb{\e_j}{\e_j}\>\zzz\right) = |\bk{\e_j}\zz|^2,\] assuming as
usual that $\zz$ is a unit vector. The image of $\mu$ is the standard simplex
\be{De}\ts \De = \De_{d-1} =
\big\{(\x_0,\ldots,\x_{d-1}):\suml_{i=0}^{d-1}\x_i=1,\
\x_i\ge0\big\}\ee of $\R^d$. (In future, we shall omit the subscript $d-1$ when
the context makes the dimension of the simplex clear.) The dependence on $\bC$
rather than $C$ comes from the ordering of the basis by the eigenvalues of
$\bfp$.

For any $i\in\Zd,$ we can now write
\[ \D_{i\bfp} = \sum_{j\in\Zd}\om^{ij}\kb{\e_j}{\e_j}.\]
In particular, the components of $\Ph_\zz^\bC$ relative to \eqref{sT}
are given by
\[ (\Ph_\zz^\bC)_i = \tr({\D_{i\bfp}}\zz\bmk\zz) = \sum_{j\in\Zd}\om^{ij}\mu^\bC_j.\]
This is expressed succinctly by the equation
\be{PhVmu} \Ph_\zz^\bC = V\mu^\bC(\rz)\ee
for $\zz\in S^{2d-1},$ where
\[ V = \left(\ba{cccc}
1 & 1 & \cdots & 1\\
1 & \om & \cdots & \om^{d-1}\\
1 & \om^2 & \cdots & \om^{2(d-1)}\\
\cdot &\cdot &\cdots & \cdot\\
1 & \om^{d-1} & \cdots & \om^{(d-1)^2}\\
\ea\right)\]
is the Vandermonde matrix that represents the discrete Fourier transform. The
$(i,j)$th entry of $V$ is $\om^{ij}$ (we start indexing at $0$), and $(1/\sqrt
d)V\in \U(d)$. We can summarize the discussion by the following result, a
version of which appears in \cite{FHS}:

\begin{proposition}\label{propCycMomMaps}
  For any $\bC\in\proj,$ the restricted overlap map $\Ph_\zz^\bC$ is a Fourier
  transform of the moment map $\mu^\bC$.
\end{proposition}

\begin{example}\label{concrete}\rm
The argument above is rendered more concrete by taking $\bC=\{0\}\times\Zbd,$
which is associated to $H,$ so that $\T^C$ is the standard maximal torus in
$\SU(d)$. Its moment map is determined by the diagonal entries of \eqref{Them},
so is given by setting $x_i=|\z_i|^2$ in \eqref{De}, assuming $\zz\in
S^{2d-1}$. Define \be{alpha} \al_i= \bmk\zz{h^i}\zz =
\x_0+\om^i\x_1+\cdots+\om^{i(d-1)}\x_{d-1},\ee so that
\[\left(\kern-3pt\ba{c}
  \al_0\\\al_1\u.\u.\u.\\\al_{d-1}\ea\kern-3pt\right) =
V\left(\kern-3pt\ba{c}
\x_0\\\x_1\u.\u.\u.\\\x_{d-1}\ea\kern-3pt\right),\]
and $(\al_0,\ldots,\al_{d-1})\in\sT$. 
\end{example}

For the following statements, let $C$ be a cyclic subgroup of $\Zd^2$
associated to $\bC\in\proj,$ and let
\[ \mu^\bC\colon\CPD\lra \De\]
be the associated moment mapping. Recall that correctly separated is defined in Definition \ref{cs}.

\begin{definition}\label{Cfid}
A point $\rz\in\CPD$ is \emph{$C$-admissible} if all points in its $C$-orbit
are correctly separated.
\end{definition}

\noindent It follows immediately that a unit vector $\zz$ is fiducial if and
only if $\rz$ is $C$-admissible for \emph{every} $C\in\P\Zd^2$.\smallbreak

Note that if $\rz$ is $C$-admissible, so is any point in its $\T^C$ orbit.  The
set of $C$-admissible points is therefore determined by its image under
$\mu^\bC$. A different subgroup $C'$ determines a conjugate maximal torus
$P\,\T^C\!P^{-1}$ where $P\in\rN(d)$ (see \eqref{UU}), and the new moment map
is obtained by pre-composing $\mu^\bC$ with $P$. However,

\begin{lemma}\label{adm}
The image by $\mu^\bC$ of the set of $C$-admissible points does not depend upon
$C$.
\end{lemma}

\begin{proof}
This is a consequence of \eqref{PhVmu}, in which $V$ transforms
$\De\subset\R^d$ into a subset of $\sT$. The vertices of $\De$ correspond to
the columns of $V$. A $C$-admissible point $\rz$ is characterized by the
condition that each non-identity component of \eqref{PhVmu} has norm
$1/\sqrt{d+1}$. It follows that
\[ V\A\ =\ V\De\cap\left\{\ba{ll}
\frac1{\sqrt{d+1}} T & \ifbox d\text{ is odd}\yyy
\frac1{\sqrt{d+1}}\big(T\times\{\pm 1\}\big) & \ifbox d\text{ is even},
\ea\right.\] where $T$ denotes a Clifford torus (with coordinates of unit
modulus) in $\C^n$ or in $\C^{n-1}\subset\sT$ respectively.
\end{proof}

We shall denote the universal subset of $\De$ arising in this lemma by $\A$. It
consists of the image of $C$-admissible points in $\CPD$ by the moment mapping
corresponding to the maximal torus generated by $C$. The proof of Lemma
\ref{adm} establishes a bijection between $\A$ and products of circles. In
Example \ref{concrete}, $\rz$ is $H$-admissible if and only if
$|\al_k|^2=1/(d\+1)$ for all $1\le k\le n,$ so in particular $\al_n=\pm1/\sqd$
if $d=2n$.

\begin{example}\rm
For $d=3,$ the simplex $\De$ is a filled equilateral triangle with vertices
$(1,0,0),(0,1,0),(0,0,1)$ in $\R^3,$ and \be{H3}\ts \A =
\Big\{\frac23\big(\cos^2\phi,\>\cos^2(\phi+\frac{2\pi}3),\>
\cos^2(\phi+\frac{4\pi}3)\big):\phi\in\left(-\frac\pi2,
\frac\pi2\right]\Big\}\ee is
  its incircle. This circle is generated by the point $[\zz_3]$ in \eqref{z3}
  as $\phi$ varies. The inverse image of each midpoint
  $(0,1,1),(1,0,1),(1,1,0)$ contains three of the nine points in
  \eqref{33}. Returning to \eqref{PhVmu}, the Fourier transform converts $\De$
  into the convex hull of the third roots of unity, and $\A$ into a circle of
  radius $1/2$.
\end{example}\smallbreak  

In order to describe more accurately the shape of $\A,$ we first define a
series of quadratic forms
\[ f_j = \suml_{i=0}^{d-1}\x_i\x_{i+j},\qquad j=0,\ldots,d-1\]
derived from \eqref{alpha}. To make sense of the right-hand side, the range of
indices is extended cyclically, so that $\x_i$ is defined to be equal to
$\x_{i-d}$ if $d\le i\le 2d-1$. In particular,
\[ f_0=\suml_{i=0}^{d-1}\!x_i^2, \]
and $f_{d-j}=f_j$ for $1\le j\le d-1$. If $d$ is even then $f_n=2f_n'$ where
\[ f_n' = \suml_{i=0}^{n-1}\x_i\x_{i+n}.\]
The forms $f_0,\ldots,f_n$ constitute a basis of the space $S^2(\R^d)^*$
of bilinear forms invariant by the action of $\Zd$ cyclically permuting
the $\x_i$.

We can now state

\def\SQRT#1{\sqrt{\kern-3pt\smash{#1}\vphantom{1^1_1}}}

\begin{theorem}\label{main}
Let $d\ge3$. The set $\A$ is the intersection of $\De$ with $n$
quadrics in $\R^d,$ and lies in a round sphere $S^{d-2}$ of radius
$\SQRT{\frac{d-1}{d(d+1)}}$ centred in $\De$. Moreover,
\begin{itemize}[itemsep=5pt,topsep=5pt]
\item if $d=2n+1,$ then $\A=\De\cap T$ where $T$ is a torus of revolution of
  dimension $n$ and radius $\SQRT{\frac2{d(d+1)}}$ in $S^{d-2}$;
\item if $d=2n$ then $\A=\De\cap(T'\sqcup T'')$ where $T',T''$ are 
  tori of revolution of dimension $n-1$ and radii $\SQRT{\frac2{d(d+1)}}$
  in parallel hyperspheres of $S^{d-2}$.
\end{itemize}
\end{theorem}

\noindent By a torus of revolution of radius $r,$ we mean a Euclidean product
of circles each of radius $r$. In Example \ref{d=4} below, Figure~1 displays
the sphere $S^{d-2}$ for $d=4,$ and this exits the tetrahedron (whose front
face has been removed).\smallbreak

\begin{proof}
We will use $C=H$ to compute $\A$.
Since
\[ |\al_j|^2 = \sum_{i=0}^{d-1} \om^{ij}f_i,\]
the $H$-admissible assumption also implies that
\[ V\left(\kern-3pt\ba{c}
f_0\\f_1\u.\u.\u.\\f_{d-1}\ea\kern-3pt\right) = \frac1{d+1}
\left(\kern-3pt\ba{c}d\+1\\1\u.\u.\u.\\1\ea\kern-3pt\right).\]
The sum of the entries of the $(i+1)$st row of $V$ equals
\[ \sum_{j=0}^{d-1}\om^{ij} = \frac{1-(\om^i)^d}{1-\om} = 0\]
for $1\le i\le d-1$. Therefore the unique solution
to the matrix equation must be given by
\be{must} \half f_0 = f_1 = f_2 = \cdots = f_{d-1} = \frac1{d+1}.\ee
This shows that $\A$ lies in the intersection of the sphere $S^{d-1}$ defined
by $f_0=2/(d+1)$ and the remaining quadrics $f_i=1/(d+1)$ for $1\le i\le n$. It
also lies on the intersection of $S^{d-1}$ with the plane containing $\De,$ and
this small hypersphere has radius $r$ given by
\[  r^2 = \suml_{i=0}^d\left(\x_i-\frac1d\right)^2 =
\frac2{d+1}-\frac2d\suml_{i=0}^d\x_i+\frac1d = \frac{d-1}{d(d+1)},\]
in the notation of \eqref{De}, as stated.

The values of $f_i$ found above are consistent with the equation
\[ 1 = \Big(\sum_{i=0}^{d-1}\x_i\Big)^2 = \sum_{i=0}^{d-1}f_i.\]
When $d=2n$ is even, there is an analogous equation that gives new
information, namely
\[ \al_n\!^2 = \Big(\sum_{i=0}^{d-1}(-1)^i\x_i\Big)^2 =
\sum_{i=0}^{d-1}(-1)^if_i.\]
It follows that
\[ \sum_{i=0}^{d-1}(-1)^i\x_i = \pm\frac1{\sqd}\]
and $\A$ lies in the union of two hyperplanes.

Let $V',V''$ be the real and imaginary parts of $V,$ so that $V'$ is a
matrix of cosines and $V''$ a matrix of sines. Recalling that
$\al_{d-k} = \ol{\al_k},$ set
\[ \sqd\,\al_k = \cos\th_k+i\sin\th_k,\qquad 1\le k\le n.\]
If $d$ is odd then the angles are unconstrained, but if $d$ is even then
$\th_n=0,\pi\mod{2\pi}$ to ensure that $\al_n=\pm1$. Since $d\,V^{-1} =
\ol V = V'-iV''$,
\[ d\sqd\left(\kern-3pt\ba{c}
\x_0\\\x_1\u.\u.\u.\\\x_{d-1}\ea\kern-3pt\right) = V'\!\left(\kern-3pt\ba{c}
\sqd\\\cos\th_1\u.\u.\u.\\\cos\th_1\ea\kern-3pt\right) +
V''\!\left(\kern-3pt\ba{c}
0\\\sin\th_1\u.\u.\u.\\-\sin\th_1\ea\kern-3pt\right).\] Note that the last
column vector has a zero sub-middle entry if $d$ is even since $\Im\al_n=0$. It
follows that \be{xcs} d\sqd\,\x_k = \sqd + \suml_{i=1}^{d-1}(c_{ki}\cos\th_k +
s_{ki}\sin\th_k),\ee where $c_{ki}=\cos(\frac{2\pi ki}d)$ and
$s_{ki}=\sin(\frac{2\pi ki}d)$. There are two cases to consider, according to
the parity of $d$ and the properties of $\cos\th_k,\sin\th_k$ that will reflect
\eqref{sT}.

\begin{itemize}
\item{\bf Case \boldmath$d=2n+1$.} Let $P$ be the $d\times d$ matrix indexed
  with $(i,j)\in\Zd^2$ and
\[ P_{ij}=\left\{ \ba{ll}
\sqrt2 &    \ifbox j=0\\
2c_{ij}    & \ifbox 0<j\le n\\
2s_{i(j-n)} & \ifbox n<j.
\ea\right.\]
Thus, every entry in the first column of $P$ is $\sqrt2,$ and the remaining
entries in the first row of $P$ are $2$ ($n$ times) followed by $0$ ($n$
times). Use of the Dirichlet kernel
\[ \sum_{k=-n}^n e^{k\th} =
\frac{\sin((n+\frac12)\th)}{\sin(\frac12\th)}\]
and elementary trigonometric identities imply that the rows of $P$ are
orthogonal, and that the norm squared of each one equals $2d$. Therefore,
$(1/\!\sqrt{2d})P\in\mathrm{O}(d)$ and \eqref{xcs} implies that
\[  d\sqd\!\left(\kern-3pt\ba{c}
\x_0\\\x_1\\.\\.\\.\\.\\\x_{d-1}\ea\kern-3pt\right) = 
P\!\left(\kern-3pt\ba{c}
\sqrt{n+1}\\\cos\th_1\u.\u.\\\cos\th_n\\\sin\th_1\u.\u.\\\sin\th_n
\ea\kern-3pt\right).\]

\item{\bf Case \boldmath$d=2n$.}
We again index $P$ by $\Zd^2,$ but this time we set
\[ P_{ij}=\left\{\ba{ll}
\sqrt2 & \ifbox j=0\\
2c_{ij}    & \ifbox 0<j<n\\
-(-1)^i\sqrt2 & \ifbox n=j\\
2s_{i(j-n)} & \ifbox n<j.
\ea\right.\]
Once again, $(1/\!\sqrt{2d})P\in\mathrm{O}(d)$. Equation \eqref{xcs} translates
into
\[ d\sqd\!\left(\kern-3pt\ba{c}
\x_0\\\x_1\\.\\.\\.\\.\\\x_{d-1}\ea\kern-3pt\right) = 
P\!\left(\kern-3pt\ba{c}
\sqrt{\vphantom{d}\smash{n+\half}}\\
\cos\th_1\u.\\\cos\th_{n-1}\\
\pm\cos\th_n/\!\sqrt2\\
\sin\th_1\u.\\\sin\th_{n-1}
\ea\kern-3pt\right).\]
\end{itemize}

\noindent In both cases, $\xx$ parametrizes a Clifford type torus (or tori) of
radius $\sqrt{2d}/(d\sqrt{d+1})$.
\end{proof}

\begin{example}\label{d=4}\rm
The case $d=4$ is well understood, and the relatively simple nature of its
overlap map was thoroughly explained in \cite{Bengt}. Nonetheless, the
underlying geometry of circles and golden ratios, described in \cite{Lora},
illustrates the moment map approach. Let $\zz=(\z_0,\z_1,\z_2,\z_3)$ be a unit
vector in $\C^4$ whose Heisenberg orbit $\Z_4^2\cdot\rz$ is a SIC-POVM in
$\CP^4$. Let $\x_i=|\z_i|^2$ and abbreviate $\Ph_\zz$ to $\Ph$. Then
\[\ba{ccccl}
1 &=& \Ph(0,0) &=& \x_0+\x_1+\x_2+\x_3\y
\frac1\5e^{i\th} &=& \Ph(0,1) &=& \x_0+i\x_1-\x_2-i\x_3\y
\pm\frac1\5 &=& \Ph(0,2) &=& \x_0-\x_1+\x_2-\x_3,
\ea\]
for some $\th\in \U(1)$ and choice of sign. It follows that
\be{sitau}
2\5\,\xx = \left\{\ba{l}
\big(\varphi+\cos\th,\>\psi+\sin\th,\>\varphi-\cos\th,\>\psi-\sin\th\big),
\hbox{ or}\y
\big(\psi+\cos\th,\>\varphi+\sin\th,\>\psi-\cos\th,\>\varphi-\sin\th\big),
\ea\right.\ee
where $\varphi=\half(\5+1)$ and $\psi=\half(\5-1)$.

This shows that $\xx$ must belong to the disjoint union of two circular arcs
each of radii $1/(2\5)$ suspended in the hyperplane $\sum\x_i=1$ of $\R^4,$ as
in Figure~1. However, the circles themselves escape the confines of the moment
polytope \eqref{De} (here a solid tetrahedron) that is the image of $\CP^3$. A
related figure appears in \cite[section 6]{ABBGGL}, which combines the moment
map approach with the use of special `spinor' bases in the situation in which
(as here) $d$ is a perfect square.

Change notation so that
$\zz=(ae^{i\alpha},be^{i\beta},ce^{i\ga},de^{i\de})$
with $\x_0\=a^2,\>\x_1\=b^2,\>\x_2\=c^2,\>\x_3\=d^2$ for clarity (so $d$ is
temporarily not a dimension). Then
\[ 0 = |\Ph(2,0)|^2-|\Ph(2,2)|^2 =
16abcd\cos(\alpha-\gamma)\cos(\beta-\de).\]
One can check that none of $a,b,c,d$ can vanish, meaning that (in contrast to
the case $d=3$) points of the SIC-POVM cannot project to the boundary of the
polytope. If we fix the second circle in \eqref{sitau}, we are furthermore
forced to assume that $\alpha-\gamma=\pm\pi/2$. Without loss of generality, we
can then set $\delta=0,$ which implies \be{ga}\ts
\frac15=\Ph(0,2)^2=4b^2d^2\cos^2\!\beta,\ee and
$(\varphi^2-\sin^2\th)\cos^2\beta = 2$.  This relationship can be interpreted
as a link between moment maps arising from different maximal tori.

The equations
\[ |\Ph(1,0)|^2-|\Ph(1,2)|^2 = 0 = |\Ph(1,1)|^2-|\Ph(1,3)|^2\]
allow us to eliminate $\beta$ and deduce that
\[(ad\-bc)(ad\+bc)(ab\-cd)(ab\+cd) = 4a^2b^2c^2d^2\cos^2\beta.\]
One can eliminate $\beta$ using \eqref{ga} to find that \be{csth}\ts
\cos\th=\pm\half\sqrt{3-\sqrt5}=\pm\frac\psi{\sqrt2},\qquad
\sin\th=\pm\half\sqrt{1+\sqrt5}=\pm\frac{\sqrt\varphi}{\sqrt2}.\ee This gives
four possible angles around each circle \eqref{sitau}, represented by the eight
`beads' in Figure~1. If we fix one of these solutions on the second circle, we
have the following choices: 4 for $\alpha,$ for each of these 2 for $\gamma,$
and independently 4 for $\beta$. Choices of signs for $a,b,c,d$ are taken care
of by the different angles and overall phase. This gives a total of 32 fiducial
points lying over each bead, in accordance with the known results \cite{App}.

\begin{center}
\vspace{10pt}  
\scalebox{.42}{\includegraphics{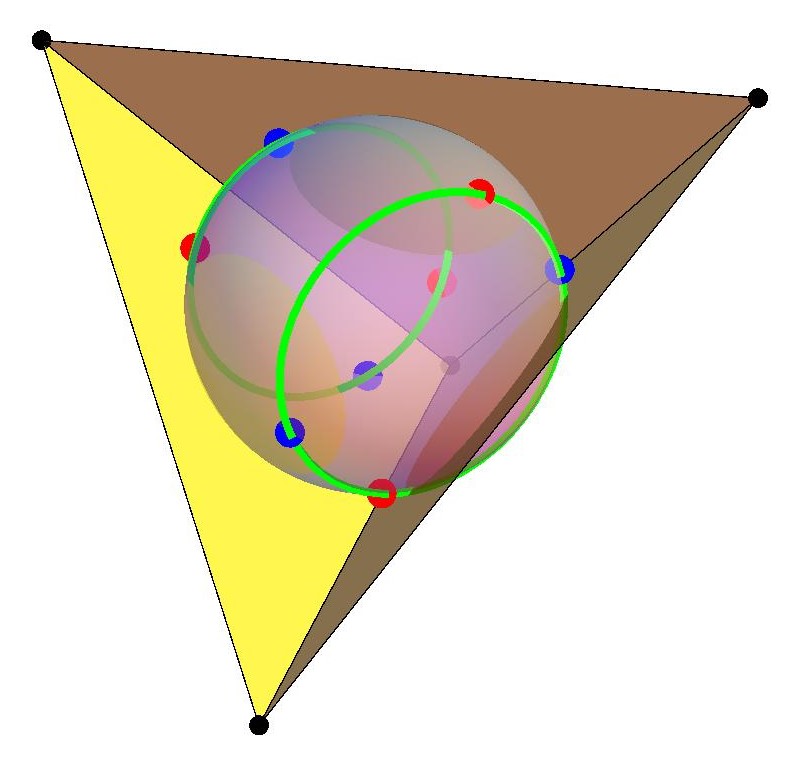}}\y
Figure~1: Fiducial images for $d=4$
\vspace{20pt}
\end{center}

The solution described by Bengtsson \cite[eq.\,(7)]{Bengt} has overlap map
\[ \Ph|\Zd^2 = {\ts\frac1\5}\!\left(\!\ba{cccc} 
\5 & u & -1 & \ol u\\ 
u  & \ol u & -\ol u & \ol u\\
-1  & -u & -1 & \ol u\\
\ol u & u & u & u
\ea\!\right),\]
where
\[\ts u=e^{i\th}=\sqrt5\,\Ph(0,1)=\frac1\2(\psi+i\,\sphi)\]
in accordance with \eqref{csth} and \cite[eq.\,(8)]{Bengt}. Its associated
fiducial vector $\zz$ (unique up to phase) can be recovered from the first
column of \eqref{recover}. Set $\rho=1+\sqrt2$ and
$x=\tan\th=\sphi/\psi=\sqrt{2+\sqrt5}$. Then
\[ \rz = \Big[2\rho,\ 1+\rho x + i(\rho+x),\ 2i,\
1-\rho x+i(\rho-x)\Big],\] and $\rz$ projects to a point on the second
circle in \eqref{sitau}. There are however simpler ways to express such a
fiducial vector, see \cite{Bend,IB}.

Let $\K=\QD=\Q(\5)$ as in the Introduction, noting that this is a subfield of
$\E_1=\Q(e^{i\th})$. The minimum polynomial of $e^{i\th}$ over $\Q$ has degree
$8,$ whereas its splitting field $\E=\E_1(i)$ has degree 16 over $\Q$. Indeed,
the Galois group $\Gal(\E/\Q)$ is isomorphic to $\Z_2\times D_8,$ where $D_8$
is the dihedral group, and $\K$ is the fixed field of the normal subgroup
$\Z_2\times V_4$. Therefore, $\E$ is an \emph{abelian} extension of $\K,$ and
$\E$ contains a fiducial vector $\zz.$ However, all the information of the
SIC-POVM is provided by $e^{i\th},$ which is a unit in the smaller field
$\E_1$. In Figure~1, the $8$ beads form an orbit of $D_8$ acting as rotations
of the tetrahedron.
\end{example}\medbreak

Lemma \ref{adm} and Theorem \ref{main} help one to grasp the essence of the
SIC-POVM problem. Whether $d$ is even or odd, the set of $H$-admissible points
is characterized by the equations \be{ff} f_1=\cdots=f_n=\frac1{d+1},\ee
derived from \eqref{must}, assuming the normalization $f_0=1$. They
characterize a real subvariety of $\CP^{d-1}$ of codimension $n$. For each
$C\in\P\Zd^2,$ we can consider the translate of this variety determined by an
element of the Clifford group mapping $H$ to $C$. The resulting equations were
in effect already written down by several authors \cite{ADF,FHS,Kha}. The
SIC-POVM existence question is whether these subvarieties have a non-empty
intersection as $C$ varies over $\P\Zd^2,$ whose size is determined by Lemmas
\ref{mn} and \ref{lemSizePZ} below. For example, if $d$ is prime then there are
$d+1$ subgroups and subvarieties to consider.

Given that $\CPD$ has real dimension at most $4n,$ one might conjecture that
four subgroups can always be found to reduce the set of solutions to be finite
provided $d>3$. Such a statement appears to be a weaker version of the $3d$
conjecture of \cite{FHS}, see also \cite{ADF}. An infinitesimal version of this
setup is the following. For each $\rz\in\CPD$ and $C\in\P\Zd^2,$ one can in
theory compute the tangent space to the level set of the $\T^C$-invariant real
polynomial $f_i$ at $\rz$. Knowledge of the configuration of these spaces as
$i$ and $C$ vary would have a bearing on the finiteness question for Heisenberg
SIC-POVMs, but is an independent problem that may be accessible for small
values of $d$.

\setcounter{equation}{0}
\section{Overlap symmetries and Galois conjugation}\label{secGalois}

Suppose that $\zz\in\C^d$ is a fiducial vector for a Heisenberg SIC-POVM. This
means that $\bk\zz\zz=1$ and any two points of the orbit $(W\times
H)\cdot[\zz]$ are correctly separated, recall Definition \ref{cs} and
\eqref{wh}. Throughout the ensuing discussion, we fix the unit vector $\zz$ and
introduce fields and groups that depend implicitly on $\zz$.

Define the symmetry group $S$ of $\Ph_\zz$ by
\[ S:=\{G\in\GL_2\Zbd\colon \Ph_\zz\circ G = \Ph_\zz\}.\]
Thus, $S$ consists of automorphisms of $\Zbd^2$ that have no effect on the
overlap phases. Zauner conjectured that $S$ always contains an element $F$ of
order 3 with trace $=-1,$ an example being
\[ F_z:=\begin{pmatrix} 0 & -1 \\ 1 & -1\end{pmatrix}.
\]

\begin{remark}\rm
When $d$ is even, the relations \ref{eqnDShift} show that $(d+1)I\in S,$ and
for this reason, the matrix
\[ \hat F_z:=\begin{pmatrix} 0 & d-1 \\ d+1 & d-1 \end{pmatrix}\] is often used
instead of $F_z$. When $d$ is odd, obviously $F_z=\hat F_z,$ but when $d$ is
even, $\hat F_z$ is order $6$ in $\GL_2\Zbd,$ with $\hat F_z^2=F_z^2$ and $\hat
F_z^3 = (d+1)I$.
\end{remark}

\begin{definition}[\cite{AYZ}]\label{csc}
The fiducial $\zz$ is called \emph{centred} if $S$ contains an element $F$ of
order 3 with trace $-1$ and is displacement-free, in the sense that the corresponding Clifford unitaries are of the form $U_{F,0}$ for some $F\in\GL_2\Z_\bd$. A centred fiducial $\zz$ is said to be {
  \emph{type}-$z$} if $F$ is conjugate to $F_z,$ and { \emph{type}-$a$}
otherwise.
\end{definition}

To understand the different types, first note that by
\cite{hillar2007automorphisms},
$\GL_2\Z_{\bd} \cong \prod_{p|d}\GL_2\Z_{p^{r_p}},$
where $\bd = \prod_{p|d}p^{r_p}.$ The image $F_{p^{r_p}}\in\GL_2\Z_{p^r_p}$ of
$F$ under each of the projections must also have trace $-1$ and order a factor
of $3$. When the order is $3,$ then it is conjugate to $F_z$. The only
alternative is order $1,$ so that $F_{p^{r_p}}=I$. The condition of having
trace $-1$ then forces $p^{r_p}=3$.

To summarize this, $\zz$ is always type-$z$ unless $\gcd(9,d)=3$ and $F\equiv
I\mod3$.

The centred condition allows us to understand the Galois action. Recall that
$\E_1$ is the field extension of $\K$ generated by the image of $\Ph_\zz$. Let
$\E_0$ be the maximal subfield of $\E_1$ such that the Galois group
$\G=\Gal(\E_1/\E_0)$ fixes the image of $\Ph_\zz$ as an unordered set.

\begin{theorem}[\cite{AYZ}]\label{AYZ}
Let $\zz$ be a centred fiducial. For each $g\in\G,$ there exists some
\\$G_g\in\GL_2\Zbd$ and $r_g\in\Zbd^2$ such that
$$g\circ\Ph_\zz(\bfp) = \left\{\ba{ll}
\Ph_\zz(G_g\bfp) & \ifbox 3\centernot|d \y
\si^{\langle r_g,\bfp\rangle}\Ph_\zz(G_g\bfp) & \ifbox 3|d
\ea\right.,$$
  where $\si$ is a third root of unity.
\end{theorem}

We will soon see that the above formula can be simplified if we assume a
stronger condition, which we will now work to motivate. Let $\E$ be the field
extension of $\K=\QD$ generated by the $\bd$th root of unity $\tau$ and the
components of the projector $\zzz$.

\begin{lemma}
	$\E = \E_1(\tau)$.
\end{lemma}

\begin{proof}
Note that each $\D_\bfp$ is a matrix with entries in $\Q(\tau)$. The
formula $\Ph_\zz(\bfp)=\tr(\D_{\bfp}\zzz)$ shows that $\Ph_\zz(\bfp)$
takes values in $\E,$ so that $\E_1\le \E$.
	
For the other inclusion, \eqref{recover} implies that $\zzz\in\E_1(\tau),$ so
that $\E\le \E_1(\tau)$.
\end{proof}

Working more carefully, we note that $\zzz$ being Hermitian means that
\[ \Ph_\zz(\bfp)=\tr(\D_{\bfp}\zzz)\in \Q(\zzz,\Re\tau).\]
In all of the known solutions, we actually have $\Re\tau\in\E_1$. Moreover, if
the square-free part $\hd$ of $\bd$ is congruent to $1$ modulo $4,$ then $\E_1$
contains $\Im\zeta_\hd,$ where $\zeta_\hd$ is a $\hd$th root of unity.

\begin{lemma}
  If $\E_1$ contains $\Im\zeta_\hd$ whenever $\hd\equiv1\mod4,$
  then $\E\ge\E_1(i\sqrt{\bd})$.
\end{lemma}

\begin{proof}
Since $\E = \E(\tau),$ it suffices to show $i\sqrt\bd\in\E$. By the
Kronecker-Weber theorem, for any square-free integer $n,$ $\Q(\sqrt n)\le
\Q(\zeta_{|\De|}),$ where
$$\De= \left\{\ba{ll}
n & \ifbox n\equiv 1\mod4\y
4n &\ifbox n\equiv 2,3\mod4
\ea\right.$$
is the discriminant of $\Q(\sqrt n)$. We proceed by taking cases:
\begin{itemize}[itemsep=5pt,topsep=5pt]
\item{\bf Case \boldmath$\hd\equiv 3\mod4$.} Taking $n=-\hd$ gives
  $i\sqrt\bd = \sqrt{-\hd}\in\Q(\zeta_\hd)\le\Q(\zeta_\bd)\le\E.$
\item{\bf Case \boldmath$\hd\equiv 1\mod4$.} Taking $n=\hd$ gives
  $\sqrt\bd\in\E$. But we also have
  $$\E\ge \Q(\tau)\ge \Q(\zeta_\hd)\ge\Q(i\Im\zeta_\hd),
  \qbox{and}\E\ge\E_1\ge\Q(\Im\zeta_\hd),$$
so that $i\in\E$. Thus $i\sqrt\bd\in\E$.
\item{\bf Case \boldmath$2|d$.} Then $4|\bd$. Taking $n$ to be the negative
  of the square-free part of $\bd/4,$ we find that $\Q(i\sqrt\bd) =
  \Q(\sqrt n)\le \Q(\zeta_{4n})\le\Q(\zeta_\bd)\le \E$.
\end{itemize}
The proof is complete.\end{proof}

This motivates

\begin{definition}[\cite{appleby2018constructing}]\label{scentred}
  A centred fiducial is called \emph{strongly centred} if the image of
  $\Ph_\zz$ generates a field extension $\E_1$ of $\K$ for which $\E =
  \E_1(i\sqrt\bd).$
\end{definition}

\noindent One consequence of being strongly centred is that the $r_g$ from
Theorem \ref{AYZ} vanishes. By \cite{appleby2018constructing}, this gives a
cleaner formula with no condition on $d$:
\be{defG}g\circ\Ph_\zz = \Ph_\zz\circ G_g.\ee
Every known fiducial is equivalent under the action of the Clifford group to a
strongly centred one. \smallbreak

In Theorem \ref{AYZ}, $G_g$ is well defined up to multiplication by an element
of $S$. It follows that there exists a subgroup $M$ of $\GL_2\Zbd$ in the
centralizer $C(S)$ of $S$ such that $\G\cong M/S,$ with the isomorphism given
by $g\mapsto G_gS$. All known solutions satisfy the following

\begin{conjecture} [{\rm\cite{appleby2018constructing}}]\label{conjM}
$M$ is a maximal abelian subgroup of $C(S)$.
\end{conjecture}

\noindent{\em For the remainder of the paper, we shall assume this conjecture and that
  $\zz$ is strongly centred.}\medbreak
  
By Theorem \ref{AYZ} and the definition of $S,$ the right action of $M$ on
$\Ph_\zz$ relates the overlap phases by (possibly trivial) Galois
conjugation. For a fixed $d,$ different fiducials may have different symmetry
groups. However, $M$ is determined only by $d$ and the type of the fiducial.

Assuming Conjecture \ref{conjM}, we can compute the structure of the group $M$. The first step has been taken in the type-$z$ case:

\begin{lemma} [{\rm\cite{AYZ}}] \label{lemCentZ}
  For type-$z$ fiducials,
$M = C(S)=\Zbd[I,F]^\times.$
\end{lemma}

\noindent Here $\Zbd[I,F]^\times\le \GL_2\Zbd$ are the set of invertible
elements in the algebra $\Zbd[I,F]$.\smallbreak

We will relate this lemma to number theory as follows.  Let $u_f$ be the
fundamental unit in $\OK$ that is positive with respect to the real embedding
of $\K$ for which $\sqrt D$ is positive, which we label $\infty_1$. Let $u_D$
be the smallest power of $u_f$ with norm $1$ (so that $u_D$ is either $u_f$ or
$u_f^2,$ since the norm of $u_f$ is in $\{\pm1\}$).  By \cite{AFMY}, there
exists some $r\in\N$ such that the rational part of $u_D^r$ equals $(d-1)/2,$
and $u_D\mod\bd$ has order $3\iota r$ where $\iota=\bd/d$.
 
\begin{lemma}\label{criteria}
The homomorphism $\iot:\Zbd[I,F]\to\OK/(\bd)$ that maps $I$ to $[1]$ and $F$
to $[u_D^{\iota r}]$ is an isomorphism if and only if either $d$ is coprime to $3$
or $3|D$ and $d\centernot\equiv 3\mod{27}$.
 \end{lemma}

\begin{proof}
Note that $j$ is well defined since $F$ and $u_D^{\iota r}$ are both order $3$.
Recall that \\$\OK=\Z[1,\om],$ where $\om\in\{\sqrt D,\frac{1+\sqrt D}2\}.$ The
surjectivity of $\iot$ is equivalent to $[\om]$ being in the image of
$\iot$. This is equivalent to the $\om$-coefficient $y\in\Z$ of $u_D^{r}$ being
invertible modulo $\bd$. This equivalence is clear when $d$ is odd, since
$u_D^r=\iot(F)$. When $d$ is even, then $\iot(F)=\left(u_D^r\right)^2$ has
$\om$-coefficient $y(d-1),$ since the rational part of $u_D^r$ is
$(d-1)/2$. The equivalence then follows from $d-1$ being invertible modulo
$\bd$.

Since $y\ne0,$ $\iot$ can only have a kernel when $d$ and $y$ share a common
factor, which is the condition for $\iot$ to not be surjective. Thus $\iot$ is
bijective if and only if it is surjective.
	
Since $u^r$ has norm $1$ and rational part $\frac{d-1}2,$ the norm of $2u^r$ is
$4=(d-1)^2-D\si^2y^2,$ where $\si^{-1}$ is the $\sqrt D$ coefficient of
$\om$. Modulo $\gcd(d,y),$ this gives $4\equiv 1\mod{\gcd(d,y)},$ so that
$\gcd(d,y)|3$. In particular, $\iot$ is not surjective if and only if
$\gcd(d,y)=3$.
	
If $3|\gcd(D,y),$ then the norm of $2u_D^r$ modulo $27$ gives $3\equiv
d^2-2d\mod{27},$ so that $d\equiv -1\text{ or }3\mod{27}$. Thus $3=\gcd(d,y)$
if and only if $d|3$ and either $3\centernot|D$ or $d\equiv 3\mod{27}$.
\end{proof}

\begin{remark}\rm
  Since $\gcd(d,y)|3,$ we see that $\ker\iot$ is generated by $(\bd/3)F\in
  (\bd/3)M\cong M_3$. Here we write $M_n$ to mean the projection of $M$ onto
  $\GL_2\Z_n$ when $n$ is a factor of $\bd$ co-prime to $\bd/n$. If $\zz$ is
  type-$a,$ then the kernel is trivial, since the same is true of the
  projection of $F$ to $M_3$.
\end{remark}

The map $\iot$ restricts to a homomorphism
$\iot^\times\colon\Zbd[I,F]^\times\to\left(\OK/(\bd)\right)^\times$ with the
same isomorphism criteria as in Lemma \ref{criteria}. Whenever these criteria
  are satisfied there are only type-$z$ fiducials (in the known
  solutions). These satisfy $\left(\OK/{(\bd)}\right)^\times\cong
  \Zbd[I,F]^\times\cong M,$ where the last isomorphism is Lemma \ref{lemCentZ}.

  Consider the case when $\zz$ is a type-$a$ fiducial so that $\iot$ is not an
  isomorphism. By the type-$a$ condition and Chinese remainder theorem, we have
  $\Zbd[I,F]=\Z_3\times\Z_{\bd/3}[I,F]$. This gives
  
\[\ba{rcl} M 
&\cong& M_3\times M_{\bd/3}\y
&\cong& (M_3/\Z_3^\times\times\Z_3^\times)\times\Z_{\bd/3}[I,F]^\times\y	
&\cong& M_3/\Z_3^\times\times\Zbd[I,F]^\times\y
&\cong& M_3/\Z_3^\times\times
\left(\OK/(\bd)\right)^\times\kern-5pt/\mathrm{coker}\iot^\times,
\ea\]
where $\Z_3^\times$ acts by scalar multiplication, and we used
$M_{\bd/3}\cong\Z_{\bd/3}[I,F]^\times$ from Lemma \ref{lemCentZ}, since $F$
modulo $\bd/3$ is congruent to $F_z$. This gives the exact sequence
$$1\to M_3/\Z_3^\times \to M\to\left(\OK/{(\bd)}\right)^\times\to\mathrm{coker}\iot^\times\to 1.$$

In the known solutions, the type-$a$ fiducials satisfy $M\cong
\left(\OK/(\bd)\right)^\times,$ which by the above exact sequence is equivalent
to $M_3/\Z_3^\times\cong\mathrm{coker}\iot^\times$ or equivalently $M_3\cong
\left(\OK/(3)\right)^\times$. Assuming the existence of this last isomorphism, which we will denote by $\iot'$, allows us to extend $\iot^\times$ to an isomorphism
$$\hot = (\iot',\iot^\times): M = M_3\times \Z_{\frac\bd3}[I,F]^\times\cong\left(\OK/(3)\right)^\times\times\left(\OK/\big(\tfrac\bd3\big)\right)^\times\cong\left(\OK/(\bd)\right)^\times.$$
In the type-$z$ case, we will let $\hot = \iot^\times :M\to\left(\OK/(\bd)\right)^\times.$
This discussion motivates

\begin{definition}
A strongly centred fiducial is \emph{algebraic} if $\hot:M\to
\left(\OK/{(\bd)}\right)^\times$ is an isomorphism.
\end{definition}

Note that $\left(\OK/(\bd)\right)^\times$ was computed in \cite{ranum1910group}, but we need to do
work before applying this result to our setting.
 
\begin{lemma}\label{lemPrimeTypes}
Let $p$ be a prime factor of $d$. If $p\equiv 1\mod3,$ then $p$ splits in
$\K$. If $p\equiv 2\mod3,$ then $p$ is prime in $\K$. If $p=3,$ then $p$
splits/ramifies/is prime in $\K$ if $D\equiv 1/0/{-1}\mod3$.
\end{lemma}

\begin{proof}
Since $\K=\QD$ is a quadratic extension, $p$ is prime (respectively
ramifies or is a product of two primes) in $\K$ if and only if the
polynomial
$$f(x):=\left\{\begin{matrix} x^2-D & \ifbox D\centernot\equiv
1\mod4 \\ (2x-1)^2-D\ & \ifbox D\equiv 1\mod4
	\end{matrix}\right.$$
is irreducible (respectively a square, or a product of different linear
factors) modulo $p$ \cite{conrad2014factoring}. When $p=3,$ the claim is
immediate. Otherwise, $p$ being prime is equivalent to $D$ not being a
quadratic residue modulo $p$. Since $D$ is the square-free part of
$(d+1)(d-3),$ this is equivalent to $(d+1)(d-3)\equiv -3 \mod p$ not being a
quadratic residue modulo $p$. By quadratic reciprocity, this is equivalent to
$p$ not being a quadratic residue modulo $3,$ so that $p\equiv 2\mod3$.
	
On the other hand, $p$ ramifies if and only if $p|D,$ which can only happen
if $p=3$ since $D|(d+1)(d-3)\equiv -3 \mod p$. Thus if $p\equiv 1\mod3,$ then
$p$ splits.
\end{proof}

\begin{lemma}
	If $9|d,$ then $D\centernot\equiv 3 \mod9$.
\end{lemma}
\begin{proof}
  Assume that $9|d$ and $D\equiv 3 \mod9$.  $u_D^r$ is a solution
  $\frac{T+U\sqrt D}2$ of the Pell equation $T^2-DU^2=4$. We know that the
  rational part of $u_D^r$ is $\frac{d-1}2,$ so that $T=d-1$. The Pell equation
  modulo $9$ is $1-3U^2\equiv 4\mod9$. Rearranging and dividing by $3$ gives
  $U^2\equiv -1\mod3,$ a contradiction.
\end{proof}

This allows us to present the result of \cite{ranum1910group} applied to our
situation:

\begin{theorem} [{\rm\cite{ranum1910group}}] \label{thmRanum}
Let $p$ be a prime and $k\in\N$ be larger than $1$ if $p=2$. Then
$$\left(\OK/(p^k)\right)^\times\cong\left\{
\ba{ll}
\Z_{p-1}^2\times\Z_{p^{k-1}}^2 & \ifbox p>2 \text{ splits in }\K\y
\Z_{p^2-1}\times\Z_{p^{k-1}}^2 & \ifbox p>2 \text{ is prime in }\K\y
\Z_6\times \Z_{2^{k-1}}\times\Z_{2^{k-2}} & \ifbox p=2\text{ and }k>1\y
\Z_6\times\Z_{3^{k-1}}^2	& \ifbox p=3\text{ ramifies in }\K.			
\ea\right.$$
\end{theorem}

\noindent For this statement of the result we needed the previous lemma since
\cite{ranum1910group} treats separately the cases when $3|d$ and $D\equiv
3\text{ or }6\mod9$. The lemma allows us to deduce that $k=1$ when $D\equiv
3\mod9,$ and notice that both cases give the same result when $k=1$.
\smallbreak

In \cite{appleby2018constructing}, type-$a$ fiducials were labelled by types
$a_4,$ $a_6$ or $a_8$ corresponding to $M_3$ being isomorphic to $\Z_2^2,$
$\Z_6$ or $\Z_8$ respectively.

We can now prove Proposition \ref{propTypeCriteria}:
\begin{proof}[Proof of Proposition \ref{propTypeCriteria}]
	By Lemmas \ref{lemCentZ} and \ref{criteria}, algebraic fiducials which are type-$z$ satisfy $d$ is coprime to $3$ or $3|D$ and $d\centernot\equiv 3\mod{27}$. Thus type-$a$ fiducials satisfy $3|d$ and either $3\centernot | D$ or $d\equiv 3\mod{27}$. When $d\equiv 3\mod{27},$ then $D\equiv 0\mod 3$. By combining Lemma \ref{lemPrimeTypes} with the previous theorem, we find that $D\equiv 0 / 1 / 2$ respectively gives $M_3\cong \Z_6 / \Z_2^2 / \Z_8$ respectively, which gives the claim.
\end{proof}

In the known examples, every strongly centered fiducial which is not algebraic is type-$z$ despite
type-$a$ solutions existing for the given $d$. Since algebraic type-$z$
solutions with $3|d$ have $3$ ramified, we can think of non-algebraic solutions
as having $3$ pretending to be ramified.

Now that we have more understanding of the structure of $M,$ we will study its
action on $\Zbd^2$. First note that the center $Z$ of $\GL_2\Zbd$ is contained
in $M$. Its action preserves each $C\in\proj,$ so we get an induced action of
$M/Z$ on $\proj$. Before we study this action, first note that we have a
Chinese remainder theorem for $\proj$:

\begin{lemma}\label{mn}
If $2\ge n,m\in\N$ are coprime then $\P\Z_{mn}^2\cong\P\Z_m^2\times\P\Z_n^2$.
\end{lemma}

\begin{proof}
Note that $\P\Z_{mn}^2$ consists of cyclic subgroups whose generators lie in
the set
\[\Z_{mn}^{2\times}:=\left\{\bfp\in\Z_{mn}^2:\langle p_1,p_2\rangle = \Z_{mn}\right\}.\]
Using the Chinese remainder theorem 
$(\pi_m,\pi_n)\colon\Z_{mn}\stackrel\cong\to\Z_m\times\Z_n,$ we find
\[\langle p_1,p_2\rangle = \Z_{mn} \iff \langle \pi_mp_1,\pi_mp_2\rangle
= \Z_{m}\text{ and }\langle \pi_np_1,\pi_np_2\rangle = \Z_{n}.\]
Thus $\Z_{mn}^{2\times} \cong \Z_m^{2\times}\times\Z_n^{2\times}$. Since
$\Z_{mn}^\times\cong\Z_m^\times\times\Z_n^\times,$ where the $\times$ in the
exponent denotes the group of units, we have
\[
	\P\Z_{mn}^2 =\Z_{mn}^{2\times}/\Z_{mn}^\times 
	\cong \Z_m^{2\times}/\Z_m^\times 
		\times\Z_n^{2\times}/\Z_n^\times 
	= \P\Z_m^2\times\P\Z_n^2
\]
as stated.
\end{proof}

This allows us to reduce to the case when $\bd= p^k,$ which one can easily
count:

\begin{lemma}\label{lemSizePZ}
  $\big|\P\Z_{p^k}^2\big| = p^{k-1}(p+1).$
\end{lemma}

\begin{proof}
If {\small$\begin{pmatrix} a\\b \end{pmatrix}$} generates some
$C\in\P\Z_{p^k}^2,$ then either $a,$ $b,$ or both lie in
$\Z_{p^k}^\times$. This gives
\[ \Z^{2\times}_{p^k} =
\big(\Z_{p^k}^\times\times\ol{\Z_{p^k}^\times}\big)
\sqcup\big(\ol{\Z_{p^k}^\times}\times\Z_{p^k}^\times\big)
\sqcup\big(\Z_{p^k}^\times\times\Z_{p^k}^\times\big),\]
where $\ol{\Z_{p^k}^\times} = \Z_{p^k}\backslash\Z_{p^k}^\times$. Thus
\[\left|\P\Z_{p^k}^2\right| 
= 2\big|\ol{\Z_{p^k}^\times}\big| + \big|\Z_{p^k}^\times\big|
= 2\big|\Z_{p^k}\big| - \big|\Z_{p^k}^\times\big| = 2p^k-\phi(p^k) = 2p^k-p^{k-1}(p-1),\]
as required.
\end{proof}

Using the Chinese remainder theorem for $M,$ it suffices to consider the action
of $M_{p^k}/Z_{p^k}$ on $\P\Z_{p^k}^2$.

\begin{lemma}
For algebraic fiducials,
\[ M_{p^k}/Z_{p^k}\cong\left\{\ba{ll}
		\Z_{p-1}\times\Z_{p^{k-1}}	& \ifbox p\text{ splits}
	\\	\Z_{p+1}\times\Z_{p^{k-1}}\	& \ifbox p>2\text{ is prime}
	\\	\Z_6\times\Z_{2^{k-2}}	& \ifbox p=2
	\\	\Z_3\times\Z_{3^{k-1}}	& \ifbox p=3\text{ ramifies}.
\ea\right.\]
For non-algebraic fiducials, the result is the same except when $p=3,$ where
$M_{3}/Z_{3}\cong$ $\Z_3$ (as if $3$ ramifies).
\end{lemma}

\begin{proof}
It is well known that 
\[ Z_{p^k}\cong\Z_{p^k}^\times\cong\left\{\ba{ll}
		\Z_{p-1}\times\Z_{p^{k-1}}	& \ifbox p\ne2
	\y	\Z_2\times\Z_{2^{k-2}}\	& \ifbox p=2,\ k>1.
\ea\right.\]
Aside from the case $p=2<k,$ the result follows directly from Theorem
\ref{thmRanum}. For $p=2<k,$ there are two possible quotients of $M_{2^k}$ by
different embeddings of $\Z_2\times\Z_{2^{k-2}}$. By Theorem \ref{thmRanum},
$M_{2^k}\cong\Z_3\times\Z_2\times \Z_{2^{k-1}}\times\Z_{2^{k-2}}$. Note that
$(I+2F)^2 =-3I$ generates $Z_{2^k}/{\pm1}\cong \Z_{2^{k-2}}$. Thus $I+2F$
generates the $\Z_{2^{k-1}}$ factor of $M_{2^k}$. It follows
that $$M_{2^k}/Z_{2^k}\cong\Z_3\times \Z_{2^{k-1}}/\Z_{2^{k-2}}\times
\Z_{2^{k-2}}\cong\Z_6\times\Z_{2^{k-2}}.$$
\end{proof}

\begin{lemma}\label{lemPrimeOrbit}
The (cyclic) action of $M_{p^k}/Z_{p^k}$ on $\P\Z_{p^k}^2$ has one free orbit
and $s$ orbits of size $p^k/\bp,$ where
$$ s = \left\{\ba{ll}
0\	& \ifbox p\equiv 2\mod3\text{ or }p=3\text{ is type-}a_8\y
1	& \ifbox p= 3\text{ is type-}z \text{ or type-}a_6\y
2 	& \ifbox p\equiv 1 \mod3\text{ or }p=3\text{ is type-}a_4.
\ea\right.$$
\end{lemma}

\noindent Note that in the algebraic case, $s$ is the number of proper factors
of $p$ in $\OK$.

\begin{proof}

Combining the two previous lemmas gives that $M_{p^k}/Z_{p^k}\cong
\Z_{|\prop|-s}\times\Z_{p^k/\bp}$. We first consider the case $p^k=\bp$.
Assume that there exists a non-free orbit $o$ with more than one element.
Since the fixed points of the action of some $gZ_\bp\in M_\bp/Z_\bp$ correspond to
eigenspaces of $g,$ there can be at most $2$ of them unless
$g\in Z_\bp$.
 
Since $o$ must be fixed by some non-trivial element of $M_\bp/Z_\bp,$ this
means that $o$ must have $2$ elements.
\begin{itemize}[itemsep=5pt,topsep=5pt]
\item{\bf Case \boldmath$s<2$.} Note that $|M_\bp/Z_\bp|+|o|>|\P\Z_\bp^2|,$ so
  there is not enough room for any free orbits. Thus every orbit has size $1$
  or $2$. This contradicts $M_\bp/Z_\bp$ being a cyclic group of order $>2$
  acting effectively.
\item{\bf Case \boldmath$s=2$.} From \cite{App}, we know that $M_\bp$ is
  diagonalizable. The claim easily follows.
\end{itemize}
	
Now consider the case when $p^k>\bp$. The multiplication map $\Z_p^{k}\to
(p^k/\bp)\Z_{p^k}\cong\Z_\bp$ induces a map $\rho\colon \P\Z_{p^k}^2\to\prop$
where all of the fibers have the same size, which must be $p^k/\bp$ by Lemma
\ref{lemSizePZ}.
	
The previous lemma shows that $M_{p^k}/Z_{p^k}\cong M_\bp/Z_\bp\times
\Z_{p^k/\bp}$. One can easily show that the second factor is generated by $I+\bp F$. Since this is
congruent to $I$ modulo $\bp,$ $I+\bp F$ acts trivially on $\prop$. Thus it
restricts to an action on each fibre of $\rho,$ which have size
$p^k/\bp$. Since this number is a prime power and the same as the order of
$I+\bp F,$ it must act transitively.
\end{proof}

\begin{theorem}\label{thmOrbitFactor}
For algebraic fiducials, each $M$-orbit in $\Zbd^2$ has stabilizer
$\hot^{-1}(\ker\pi_\fn)$ for some factor $\fn$ of $\bd$ in $\OK,$ where
$\hot:M\cong(\OK/(\bd))^\times$ comes from the algebraic condition, and
$\pi_\fn\colon(\OK/(\bd))^\times \to(\OK/(\fn))^\times$ is the map which takes
elements modulo $\fn$. This gives a one-to-one correspondence between orbits of
$M$ and factors of $\bd$ in $\OK$. For non-algebraic fiducials, the same is
true if $3$ is treated as if it ramifies in $\OK$.
\end{theorem}

\begin{proof}
The action of $\GL_2\Zbd$ on $\Zbd^2$ preserves the
function
\[\ts f\colon\Zbd^2\lra\N,\quad\bfp\mapsto \gcd(p_1,p_2,\bd).\]
The action of $Z$ restricted to each $C\in\proj$ is equivalent to
the action of $\Zbd^\times$ on $\Zbd,$ whose orbits are the fibres of
$f|_C$. Thus the orbits of $M$ correspond to the intersection of the orbits of
$M/Z$ with the fibres of $f$.
	
For any $\fn|\bd,$ there is a minimal $n\in\N$ such that $\fn|n$. We find
that \\$\hot^{-1}(\ker\pi_{n})\le\hot^{-1}(\ker\pi_\fn)$
stabilizes $f^{-1}(\bd/n)=(\bd/n)\Zbd^2\cong\Z_n^2$.
	
Using Chinese remainder results, we can reduce to the case when $n=p^k$ is a
prime power. Using the previous lemma, we see that if $p$ is prime, then
$\fn=n$ and there is only one orbit in $f^{-1}(n),$ so this must have
stabilizer $\hot^{-1}(\ker\pi_{n})$.
	
If $p=\prod_{i=1}^2\fp_i$ factors into different primes
$\{\mathfrak{p}\}_{i=1}^2,$ then $M$ is diagonalizable, with $\hot$ mapping
each factor of the diagonalization to a factor of
$\left(\OK/(p^k)\right)^\times\cong\prod_{i=1}^2\left(\OK/{\mathfrak
  (\fp_i^k)}\right)^\times$. By Chinese remainder results, we can consider
each factor separately. Each factor is equivalent the action of
$\Z_{p^k}^\times$ on $\Z_{p^k},$ whose orbits are
$\{p^\ell\Z_{p^k}\}_{\ell=0}^k$. The result follows since the stabilizer of the
action of $\Z_{p^k}^\times$ on $p^\ell\Z_{p^k}$ is
$\ker(\mathrm{mod}\colon\Z_{p^k}^\times\to\Z_{p^{k-\ell}}^\times)$.
	
If $p$ ramifies with square root $\fp,$ then $\left(\OK/(\fp)\right)^\times
\cong\Z_p^\times$. Thus $\hot^{-1}(\ker\pi_\fp)\cong M_3/Z_3$ stabilizes
the exceptional orbit from the previous lemma. More generally, the free and
exceptional orbits in $f^{-1}(n)$ have stabilizers
$\hot^{-1}(\ker\pi_{n})$ and $\hot^{-1}(\ker\pi_{n/\fp})$
respectively.
\end{proof}

\setcounter{equation}{0}
\section{Structure of the number fields}\label{secG}
Now we will compute $\G$ using class field theory. We begin with some definitions.

\begin{definition}
Let $\F$ be a number field. A \emph{modulus} of $\F$ is a pair $\fm =
\fm_0\fm_\infty,$ where $\fm_0$ is an integral ideal of $\K$ and $\fm_\infty$
is a set of real embeddings $\F\hra\C$. For each modulus $\fm,$ there is a
corresponding \emph{ray class field} $\F(\fm)$. The \emph{Hilbert class field}
$\F(1)$ is the ray class field of the trivial modulus. The degree $[\F(1):\F]$
is the \emph{class number} $h(\F)$.
\end{definition}	

\begin{definition}
An algebraic fiducial is a \emph{ray class fiducial} if $\E_0 = \K(1)$ (the
Hilbert class field) and $\E_1 = \K(\fm_1),$ the ray class field over $\K$
with modulus $\fm_1:=(\bd)\infty_1,$ where $\infty_1$ is the real embedding of
$\K$ for which $\sqrt D$ is positive.
\end{definition}

Ray class fiducials are found in every dimension where fiducials are found. In
Theorem \ref{rcf}, which is duplicated below, we see that these have a nice
interpretation of the symmetry group $S$.

\begin{theorem}\label{thmSUnit}
For ray class fiducials, $S$ is isomorphic to the cyclic subgroup generated by
the fundamental unit in $\OK$ modulo $\bd$.\end{theorem}

\begin{proof}
There is a well-known five-term exact sequence
$$1\to \Unit_{\fm}(\F) \to \Unit(\F)\to \left(\mathcal O_\F/{\fm}\right)^\times
\to \Cl_\fm\to \Cl\to 1,$$
where $\Cl$ is the class group of some field $\F,$ $\Cl_\fm$ is the ray class group of some modulus
$\fm=\fm_0\fm_\infty,$ $\Unit(\F)$ is the unit group, and $$\Unit_\fm(\F) = \{\alpha\in
\Unit(\F):v_\fp(\alpha-1)\ge v_\fp(\fm_0), \si_i(\alpha)>0,
\forall \fp|\fm_0, \si_i\in\fm_\infty\},$$ where $v_\fp(n):=\max\{r\in\N:\fp^r|n\}$.
The class groups have the
property that \\$\Cl_\fm\cong\Gal(\F(\fm)/\F)$. Thus
$$\Cl_\fm/\Cl\cong\Gal\big(\F(\fm)/\F\big)/\Gal\big(\F(1)/\F\big)
	\cong\Gal\big(\F(\fm)/\F(1)\big).$$
Now consider a ray class fiducial, where $\F=\K$ and $\fm=\fm_1=(\bd)\infty_1$. We have the (not necessarily commutative) diagram of
short exact sequences

$$\begin{CD}
	@. 	@.	@. 1	@.
\\	@. 	@.	@.	@VV V	
\\	@.	1
	@.
	@.	\Z_2
	@.
	\\	@. 	@VV V	@.	@VV V
	\\
1
@> >>	S
@> >>	M
@> >>	\G
@> >> 1
\\	@. 	@VV V	@VV\cong V	@VV  V
\\	1
@> >> \Unit/\Unit_{(\bd)}
@> >> \left(\OK/(\bd)\right)^\times
@> >> \Cl_{(\bd)}/\Cl
@> >> 1
\\	@. 	@VV V	@.	@VV V
\\	@. 	\Z_2 @. 	@. 1	@.
\\	@. 	@VV V	@.	@.
\\	@. 	1 @. 	@.	@.
\end{CD}$$

\bigbreak

\noindent The $\Z_2$ term in the last column comes from
$$\G/\big(\Cl_\bd/\Cl\big)\cong
\Gal\Big(\K(\fm_1)/\K(1)\Big)/\Gal\Big(\K\big((\bd)\big)/\K(1)\Big)
\cong\Gal\Big(\K(\fm_1)/\K\big((\bd)\big)\Big)\cong\Z_2.$$ 
 The $\Z_2$ in the last column is generated by complex
conjugation \cite{AFMY}, which is the image of $-I\in M$. $-I$ gets mapped to
$-1$ in $\left(\OK/(\bd)\right)^\times,$ which must generate the $\Z_2$ in the
first column. In fact, $\Unit$ is generated by $u_f$ and $-1,$ so
$\Unit/\Unit_\bd$ is generated by $[u_f]$ and $[-1],$ where $[\cd]$ indicates
equivalence classes in $\K/{(\bd)}$. 

We can factor out the instances of $\Z_2$
from the diagram to get

$$\begin{CD}
1
@> >>	S
@> >>	M
@> >>	\G
@> >> 1
\\	@. 	@VV\cong V	@VV\cong V	@VV\cong  V
\\	1
@> >> \left\langle [\pm u_f]_{(\bd)}\right\rangle
@> >> \left(\OK/(\bd)\right)^\times
@> >> \Cl_{\fm_1}/\Cl
@> >> 1
\end{CD}$$

\noindent This completes the proof by noting that $\left\langle [-u_f]_{(\bd)}\right\rangle\cong \left\langle [ u_f]_{(\bd)}\right\rangle$ unless the second group has odd order, in which case $-u_f$ could not be a generator for the quotient. 
\end{proof}

To generalize the previous theorem, we note that conjecturally ray class fiducials are minimal in the sense that every strongly
centred fiducial has $\E_1$ an extension of $\K(\fm_1)$ and $\E_0$ is an extension of $\K(1)$. Assuming this, we can describe the structure of algebraic fiducials:

\begin{theorem}\label{thmExtOfRCF}

For every algebraic fiducial such that $\K(1)\le\E_0$ and
$\K(\fm_1)\le\E_1,$  $\G$ is a cyclic extension of $\Cl_{\fm_1}/\Cl,$ and  $S$ is isomorphic to some subset of $\big\langle[\pm u_f]_{(\bd)}\big\rangle$.
\end{theorem}
\begin{proof}
First we prove that $\G$ is an extension of $\Cl_{\fm_1}/\Cl$. Note that we have short exact sequences
$$\begin{CD}
	1	@> >> \G 				@> >> \Gal(\E_1/\K(1))	@> >> \Gal(\E_0/\K(1))	@> >> 1,
\\	1	@> >> \Gal(\E_1/\K(\fm_1))	@> >> \Gal(\E_1/\K(1))	@> >> \Cl_{\fm_1}/\Cl 	@> >> 1.
\end{CD}$$
These can be combined to give
$$\begin{CD}
	1	@> >> \G\cap \Gal(\E_1/\K(\fm_1))	@> >> \G 	@> >> \Cl_{\fm_1}/\Cl 	
	@> >>  \frac{\Gal(\E_1/\K(1))}{\G\cdot\Gal(\E1/\K(\fm_1))}	@> >> 1.
\end{CD}$$

Combining this with the relevant short exact sequences gives the diagram

$$\begin{CD}
	@. 	@.	@. 1	@.
\\	@. 	@.	@.	@VV V	
\\	@.	1
	@.	1
	@.	\G\cap\Gal\big(\mathbb E_1/\K(\fm_1)\big)
	@.
	\\	@. 	@VV V	@VV V	@VV V
	\\
1
@> >>	S
@> >>	M
@> >>	\G
@> >> 1
\\	@. 	@VV j_S V	@VV \hot V	@VV  V
\\	1
@> >> \left\langle [\pm u_f]_{(\bd)}\right\rangle
@> >> \left(\OK/(\bd)\right)^\times
@> >> \Cl_{\fm_1}/\Cl
@> >> 1
\\	@. 	@.	@VV V	@VV V
\\	@. 	 @.1 	@. 1	@.
\end{CD}$$

\noindent where the injectivity of $j_S$ is deduced from the
diagram. Similarly, the vertical map onto $\Cl_{\fm_1}/\Cl$ is surjective, so
that $\frac{\Gal(\E_1/\K(1))}{\G\cdot\Gal(\E1/\K(\fm_1))}$ is trivial. We also
find that $\G$ is an extension of $\Cl_{\fm_1}/\Cl$ by $\G\cap\Gal\big(\mathbb
E_1/\K(\fm_1)\big)\cong \left\langle [\pm u_f]_{(\bd)}\right\rangle/j_s(S),$ which
is cyclic.
\end{proof}

In the non-algebraic case things are more complicated, since $j^\times$ has a kernel and co-kernel. Recall that in this case, $d\equiv 3(\text{mod }9),$ with the kernel given by $M_3/Z_3\cong\mathbb Z_3,$ and co-kernel given by the cyclic group $\left(\mathcal O_{\mathbb K}/(3)\right)^\times/\mathbb Z_3^\times \cong \mathbb Z_k,$ where the algebraic fiducials are type $a_{2k}$. This gives the diagram

$$\begin{CD}
	@. 	@.1	@. 1	@.
\\	@. 	@.	@VV V	@VV V	
\\	@.	
	@.	\mathbb Z_3
	@.	\G\cap\Gal\big(\mathbb E_1/\K(\fm_1)\big)
	@.
	\\	@. 	@.	@VV V	@VV V
	\\
1
@> >>	S
@> >>	M
@> >>	\G
@> >> 1
\\	@. 	@VV V	@VV j^\times V	@VV  V
\\	1
@> >> \left\langle [\pm u_f]_{(\bd)}\right\rangle
@> >> \left(\OK/(\bd)\right)^\times
@> >> \Cl_{\fm_1}/\Cl
@> >> 1
\\	@. 	@.	@VV V	@VV V
\\	@.  	 @. \mathbb Z_k 	@. Q	@.
\\	@. 	@.	@VV V	@VV V
\\	@. 	@.1	@. 1	@.
\end{CD}$$
where $Q:=\frac{\Gal(\E_1/\K(1))}{\G\cdot\Gal(\E1/\K(\fm_1))}$. In all of the known examples, $S$ injects into $ \left\langle [\pm u_f]_{(\bd)}\right\rangle$ and $\mathbb Z_3$ injects into $\G\cap\Gal\big(\mathbb E_1/\K(\fm_1)\big)$. There are two observed cases. In the first case, $Q=1$ and $\mathbb Z_k$ is a quotient of $\left\langle [\pm u_f]_{(\bd)}\right\rangle$. In the second case, $Q\cong\Z_k$.

\begin{theorem}\label{thmOrbitField}
For an algebraic fiducial, for each $\OK\owns\fn|(\bd),$ an overlap phase
contained in the orbit of $M$ labelled by $\fn$ takes values in the fixed field
of a Galois group isomorphic to $\Gal\big(\K(\fm_1)/\K((\fn)\infty_1)\big)$.
\end{theorem}

\begin{proof}
Using the isomorphism $M\cong(\OK/(\bd))^\times,$ the elements of the orbit
labelled by $\fn$ are those which are stabilized by the subgroup
$\mathrm{stab}(\fn)$ whose elements which are congruent to $1$ modulo
$\fn$. The quotient $\big(\OK/(\bd)\big)^\times/\mathrm{stab}(\fn)\cong
(\bd/\fn)\big(\OK/(\bd)\big)^\times\cong\big(\OK/(\fn)\big)^\times$. This gives the diagram of
short exact sequences

	$$\begin{CD}
	@. 	@.	1@. 1	@.
\\	@. 	@.	@VV V	@VV V	
\\	@.	
	@. \mathrm{stab}(\fn)
	@.	\Cl_{\fm_1}/\Cl_{{\fn}\infty_1}
	@.
	\\	@. 	@.	@VV V	@VV V
\\	1
@> >> \left\langle [\pm u_f]_{(\bd)}\right\rangle
@> >> \left(\OK/(\bd)\right)^\times
@> >> \Cl_{\fm_1}/\Cl
@> >> 1
\\	@. 	@.	@VV V	@VV V
\\	1
@> >> \left\langle [\pm u_f]_{\left({\fn}\right)}\right\rangle
@> >> \left(\OK/{(\fn)}\right)^\times
@> >> \Cl_{\fn\infty_1}/\Cl
@> >> 1
\\	@. 	 @. 	@VV V 	@VV V
\\	@.	@.1	@.1
\end{CD}$$

\bigbreak

\noindent The projection modulo ${\fn}$ gives a surjection 
$\left\langle[\pm u_f]_{(\bd)}\right\rangle \to\left\langle [\pm u_f]_{{\fn}}\right\rangle$
 with some kernel $K$. We
deduce from the commutativity of the diagram that $K$ injects into
$\mathrm{stab}(\fn)$ with $\mathrm{stab}(\fn)/K\cong
\Cl_{\fm_1}/\Cl_{\fn\infty_1}\cong \Gal\big(\K(\fm_1)/\K(\fn\infty_1)\big).$
\end{proof}

This suggests that for ray-class fiducials, the Galois orbit labelled by $\fn$
takes values in the field $\K(\fn\infty_1)$.

\begin{remark}\rm
A SIC-POVM in dimension $d' = d(d-2)$ will have the same value of $D$
as one in dimension $d$. This phenomenon is studied in
\cite{appleby2017dimension}, where the authors conjecture that every
SIC-POVM $\alpha$ in dimension $d$ is related to a SIC-POVM $\beta$ in
dimension $d'$ which they call \emph{aligned}, meaning that
$\E_1^\beta$ is an extension of $\E_1^\alpha$ and the overlap phases
of $\beta$ taking values in $\E_1^\alpha$ are the squares of overlap
phases of $\alpha$. In the language of the above theorem, these
correspond to Galois orbits whose labels are divisible by $d-2 =
\frac{d'}d$.
\end{remark}

\bibliographystyle{acm}
\bibliography{eqd2020}

\flushleft\vs

Department of Mathematics, King's College London, Strand, London, WC2R 2LS, UK

\enddocument